\documentclass[runningheads]{llncs}

\usepackage{amsfonts} 
\usepackage{amsmath}
\usepackage{amssymb}
\usepackage{stmaryrd}
\usepackage{mathtools}
\usepackage{cite}
\usepackage{tikz}
\usepackage{hyperref}
\usepackage{cleveref}
\let\proof\relax\let\endproof\relax
\usepackage{amsthm}

\usepackage{tikz}
\usetikzlibrary{calc}
\usetikzlibrary {arrows.meta,automata,positioning}
\usetikzlibrary {shapes.geometric}
\usepackage{tikz-cd}
\usepackage[all]{xy} 
\tikzstyle{point}=[circle,  inner sep=2pt, fill]

\usetikzlibrary{cd}
\usetikzlibrary{automata,positioning} 

\usepackage{subcaption}
\usepackage{graphicx}
\usepackage{caption}
\newenvironment{proofs}{%
  \proof}{\endproof}

\newif\ifdraft\draftfalse
\newif\ifarxiv\arxivtrue

\ifdraft
\usepackage{todonotes}
\else
\usepackage[disable]{todonotes}
\fi

\newcommand{\bools}{\mathbb{B}}
\newcommand{\setn}[1]{\mathbf{#1}}
\newcommand{\defeq}{\coloneqq}
\newcommand{\supp}{\mathrm{supp}}
\newcommand{\nat}{\mathbb{N}}
\newcommand{\real}{\mathbb{R}}
\newcommand{\nonnegreal}{\mathbb{R}_{\geq 0}}
\newcommand{\nonnegrat}{\mathbb{Q}_{\geq 0}}
\newcommand{\pfunc}{{\mathcal{P}}}
\newcommand{\powfunc}{{\mathcal{P}_\mathrm{f}}}
\newcommand{\mcfunc}[1]{\mathcal{D}(#1 +\{\checkmark\})\times A}
\newcommand{\mcfuncproduct}[1]{\subdist(#1 +\{\checkmark\})}
\newcommand{\dfafunc}[1]{(#1\times \bools)^{A}}
\newcommand{\subdist}{\mathcal{D}_{\leq 1}}
\newcommand{\subdistcount}{\mathcal{D}_{\leq 1,\mathrm{c}}}
\newcommand{\erealfunc}{\mathcal{R}^+_{\geq 0}}
\newcommand{\erealfunctimes}{\mathcal{R}^\times_{\geq 0}}
\newcommand{\traceMC}{\subdistcount(A^{+})}
\newcommand{\traceDFA}{\pfunc(A^{+})}
\newcommand{\sets}{\mathbf{Sets}}
\newcommand{\mcdfafunc}[1]{\subdist(#1 +\{\checkmark\})}
\newcommand{\mfunc}{\mathcal{M}}
\newcommand{\mfunccount}{\mathcal{M}_c}
\newcommand{\traceMFA}{\mfunccount(A^{+})}
\newcommand{\traceNFA}{\pfunc(A^{+})}
\newcommand{\mfafunc}[1]{\mfunc(#1 +\{\checkmark\})^A}
\newcommand{\nfafunc}[1]{\powfunc({#1} + \{\checkmark\})^A}
\newcommand{\exreal}{\mathbb{R}^{\infty}_{\geq 0}}
\newcommand{\mcmfafunc}[1]{\erealfunc(#1 +\{\checkmark\})}
\newcommand{\mcnfafunc}[1]{\erealfunc(#1 +\{\checkmark\})}

\newcommand{\nnegreal}{\mathbb{R}_{\geq 0}}
\newcommand{\id}{\mathrm{id}}
\newcommand{\Nat}[2]{\mathrm{Nat}(#1, #2)}

\crefname{section}{\S\!}{Sects.}
\Crefname{section}{\S\!}{Sects.}
\crefname{proposition}{Prop.}{Props.}
\Crefname{proposition}{Prop.}{Props.}
\crefname{definition}{Def.}{Defs.}
\Crefname{definition}{Def.}{Defs.}
\crefname{appendix}{App.}{Appx.}
\Crefname{appendix}{App.}{Appx.}
\crefname{theorem}{Thm.}{Thms.}
\Crefname{theorem}{Thm.}{Thms.}
\crefname{lemma}{Lem.}{Lems.}
\Crefname{lemma}{Lem.}{Lems.}
\crefname{corollary}{Cor.}{Cors.}
\Crefname{corollary}{Cor.}{Cors.}
\crefname{algorithm}{Alg.}{Algs.}
\Crefname{algorithm}{Alg.}{Algs.}
\crefname{figure}{Fig.}{Figs.}
\Crefname{figure}{Fig.}{Figs.}
\crefname{example}{Example}{Examples}
\Crefname{example}{Example}{Examples}
\crefname{table}{Tbl.}{Tbls.}
\Crefname{table}{Tbl.}{Tbls.}

\title{A No-go Theorem\\ for Coalgebraic Product Construction\thanks{We would like to thank the anonymous reviewers for their valuable comments and suggestions, which significantly improve this article. 
The authors were supported by the ASPIRE grant No. JPMJAP2301, JST. 
M.~K.~is supported by the JST grant No.~JPMJAX25CD,
and
K.~W.~is supported by the JST grants No.~JPMJAX23CU and JPMJPR25KD. }} 

\author{Mayuko Kori\inst{1} \and
Kazuki Watanabe\inst{2,3}}
\institute{Research Institute for Mathematical Sciences, Kyoto University, Japan\and
National Institute of Informatics, Japan \and
The Graduate University for Advanced Studies (SOKENDAI), Japan }

\titlerunning{A No-go Theorem for Coalgebraic Product Construction} 

\authorrunning{M.~Kori and K.~Watanabe} 

\begin{document}

\maketitle

\begin{abstract}
Verifying traces of systems is a central topic in formal verification. 
We study model checking of Markov chains (MCs) against temporal properties represented as (finite) automata. 
For instance, given an MC and a deterministic finite automaton (DFA), a simple but practically useful model checking problem  asks for the probability of (terminating) traces accepted by the DFA, which can be computed via a product MC of the given MC and DFA and reduced to a simple reachability problem.

Recently, Watanabe, Junges, Rot, and Hasuo proposed 
\emph{coalgebraic product constructions}, a categorical framework that uniformly explains such coalgebraic constructions using distributive laws.
This framework covers a range of instances, including the model checking of MCs against DFAs. 

In this paper, on top of their framework we first present a no-go theorem for product constructions, showing a case when we \emph{cannot} do product constructions for model checking. 
Specifically, we show that there are \emph{no} coalgebraic product MCs of MCs and nondeterministic finite automata for 
computing the probability of the accepting traces. 
The proof relies on a characterisation of natural transformations between certain functors that determine the type of branching, including nondeterministic or probabilistic branching. 

Second, we present a coalgebraic product construction of MCs and multiset finite automata (MFAs) as a new instance within our framework. 
This construction addresses a model checking problem that asks for the expected number of accepting runs on MFAs over traces of MCs. 
We show that this problem is solvable in polynomial time. 

\end{abstract}

\section{Introduction}
\label{sec:intro}
For decades, model checking has been extensively studied as a verification technique for ensuring the correctness of systems or programs adhere to a given specification. 
A standard model for representing systems with uncertainties are \emph{Markov chains} (MCs)~\cite{baierKatoen}.
The behaviour of an MC is observed through a trace $w\in A^{+}$ until it reaches a target state,
at which point it terminates; throughout the paper, we consider only terminating traces.
These traces are then checked against a specification represented by a finite automaton,
determining acceptance of each trace.
Formally, 
given an MC and a finite automaton, model checking asks for the probability that traces of the MC are accepted by the automaton.
This probability is mathematically expressed as $\sigma(L) = \sum_{w\in L} \sigma(w) \in [0, 1]$, where $\sigma$  is the subdistribution of traces on the MC, and $L\in\traceDFA$ is the recognised language of the automaton (excluding the empty string). 
Model checking of MCs against such linear-time properties has been actively explored in the literature, including in~\cite{baierKatoen,Vardi85,BaierAFK18,BustanRV04}.

A well-established approach to efficient model checking algorithms is the so-called \emph{product construction}~\cite{Pnueli77,VardiW86}. 
In this approach, given an MC $\mathbb{M}$ and a finite automaton $\mathcal{A}$, 
one first constructs an equivalent finite deterministic automaton (DFA) $\mathcal{A}_d$ by determinisation. 
We then constructs a product MC $\mathbb{M}\otimes \mathcal{A}_d$, where the reachability probability to a designated target state exactly coincides with the original probability $\sigma(L)$. 
While the determinisation construction explodes the number of states exponentially in general, it has been shown that the resulting product MC with the DFA is efficiently solvable in poly-logarithmic parallel time (NC) w.r.t. the size of its underlying graph (see e.g.~\cite{BaierK00023}), which belongs to the complexity class P.

To pursue an efficient model checking algorithm, it is natural to identify conditions under which product constructions can be applied to the target model checking problem \emph{without requiring determinisation}.  
Indeed, such conditions for MCs and non-deterministic ($\omega$-regular) automata have been explored for decades.
Notable examples include products for MCs with separated unambiguous automata~\cite{CouvreurSS03},
unambiguous automata~\cite{BaierK00023}, and limit-deterministic automata~\cite{Vardi85,CourcoubetisY95,SickertEJK16}.
While these automata are not fully deterministic, they impose restrictions on their non-deterministic behaviour, placing them between deterministic and non-deterministic automata.
To the best of our knowledge, no product construction is currently known for MCs and nondeterministic finite automata (NFAs) without determinisation.
It is fair to say that this is unlikely to be possible, 
although no formal proof of impossibility has been established previously (see also~\cite[Remark 4.4]{WatanabeJRH25}).

Recently, Watanabe et al.~proposed a unified approach to product constructions for model checking~\cite{WatanabeJRH25}, establishing a structural theory for it. 
In their framework, both systems and specifications are modelled as \emph{coalgebras}~\cite{Rutten00,J2016}, and 
a generic product construction, called the \emph{coalgebraic product construction}, provides a unified method for constructing products
 via \emph{distributive laws}, natural transformations that distribute the Cartesian product over the category of sets.
 Within this abstract formulation, they introduced a \emph{correctness criterion}, which is
a simple yet powerful sufficient condition ensuring the correctness of the coalgebraic product construction---the solution of the product coincides with that of the original model checking problem.
This framework covers
a wide range of instances, including the aforementioned problem (e.g.~\cite{baierKatoen, BaierAFK18}) and the cost-bounded reachability probability~\cite{SteinmetzHB16,HahnH16}.

On top of their coalgebraic framework,
 we further explore product constructions of MCs and automata.
Our main result in this paper is a \emph{no-go theorem} showing
that {\bf \emph{no}} coalgebraic product MCs for MCs and NFAs satisfy the correctness criterion for the model checking problem that asks for the probability of accepting traces (\cref{thm:no-go-mc-nfa}).
To prove this, we characterise natural transformations between certain functors that govern branching types, including non-deterministic and probabilistic branching (\cref{thm:monoid_lambda}). 
\begin{table}[t]
  \caption{Sets that are isomorphic to $\Nat{F_A}{F_B}$ via a reasonably simple isomorphism. 
  These functors, including the covariant finite powerset functor $\powfunc$ and the multiset functor $\mfunc$, are formally defined in~\cref{eg:cmon_functor}.  
  For instance, the set $\Nat{\mathcal{M}}{\mathcal{P}_f}$
  of natural transformations $\lambda\colon \mathcal{M}\Rightarrow \mathcal{P}_f$ is isomorphic to $\mathbf{2}^{\mathbb{N}_{\geq 1}}$ with a simple isomorphism  described explicitly in~\cref{eg:mfunc_natural_trans}.
  The characterisations of $\lambda\colon \mathcal{P}_f\Rightarrow \mathcal{P}_f$ and $\lambda\colon \mathcal{M}\Rightarrow \mathcal{M}$ have been found by Dahlqvist and Neves in~\cite{DahlqvistNeves}.}
  \centering
  \begin{tabular}{c|cccc}
    $F_A \backslash F_B$ &$\mathcal{P}_f$ & $\mathcal{M}$  &$\erealfunc$ &$\erealfunctimes$\\
    \hline
    $\mathcal{P}_f$ & $\mathbf{2}$ & $\mathbf{1}$ & $\mathbf{1}$ &$\mathbf{2}$\\
    $\mathcal{M}$ &$\mathbf{2}^{\mathbb{N}_{\geq 1}}$ & $\mathbb{N}^{\mathbb{N}_{\geq 1}}$ & $\nnegreal^{\mathbb{N}_{\geq 1}}$ &$\nnegreal^{\mathbb{N}_{\geq 1}}$ \\
    $\erealfunc$ &$\mathbf{2}^{\mathbb{R}_{> 0}}$ &$\mathbf{1}$ &$\nnegreal^{\mathbb{R}_{> 0}}$ &$?$ \\
\end{tabular}
  \label{tab:table_of_isomorphisms}
\end{table}

Beyond its application to product constructions,
our characterisation of natural transformations is of independent interest. 
A summary of our characterisation is provided in~\cref{tab:table_of_isomorphisms}. 
For instance, the following isomorphism for $\Nat{\mfunc}{\subdist}$ 
exemplifies our characterisation (\cref{prop:subfunc_pm}): 
\begin{corollary}
  There is an isomorphism $\lambda^{(\_)}\colon [0, 1]^{\mathbb{N}_{\geq 1}} \to \Nat{\mfunc}{\subdist}$ given by
    \begin{displaymath}
      \lambda^b_X(f)(x) = 
      \begin{cases}
        \frac{f(x)}{\sum_{x \in X}f(x)} \cdot b\big(\sum_{x \in X}f(x)\big) &\text{if }\sum_{x \in X}f(x) > 0, \\
        0 &\text{if }\sum_{x \in X}f(x) = 0.
      \end{cases}
    \end{displaymath}
  
\end{corollary}

Second, we introduce a new model checking problem for MCs and multiset finite automata (MFAs), which asks for the expected number of accepting runs on MFAs over traces on MCs.  
We show that a unique coalgebraic product construction exists for this setting (\cref{prop:correctmcmfas,prop:lambda_uniqueness}). 
As an immediate consequence, this result provides an alternative proof of the correctness of the known product construction~\cite{BaierK00023} for MCs and unambiguous finite automata, which can be seen as a special case of MFAs. 
Lastly, we show that model checking of products of MCs and MFAs are solvable in P (\cref{prop:complexity}). 
Our proof is indeed very simple: we reduce the problem to the computation of expected \emph{multiplicative rewards} on MCs (without non-trivial bottom strongly connected components), which is solvable in P due to  the very recent result by Baier et al.~\cite{BaierCMP25}.

In summary, our contributions are as follows: 
\begin{itemize}
  \item We characterise natural transformations between functors for branching, which plays an important role in the proof of our no-go theorem (\cref{thm:monoid_lambda}). 
  \item We present a no-go theorem for coalgebraic \emph{product MCs} of MCs and NFAs (\cref{thm:no-go-mc-nfa}).
  \item We propose a new model checking problem of MCs and MFAs, and show that the problem is solvable in polynomial time (\cref{prop:complexity}).
\end{itemize}

{\bf{Structure}}. In~\cref{sec:preliminaries}, we recall preliminaries related to coalgebraic semantics of transition systems, including MCs and DFAs. 
In~\cref{sec:coprodconst}, we review the coalgebraic product constructions and the correctness criterion, which is a sufficient condition that ensures the correctness of coalgebraic products with respect to a given model checking problem~\cite{WatanabeJRH25}. 
In~\cref{sec:naturaltrans}, we show a characterisation of natural transformations between certain functors determining the type of branching. 
In~\cref{sec:mcnfas}, we show a no-go theorem for coalgebraic product MCs of MCs and NFAs. 
In~\cref{sec:mcmfas}, we study a new model checking problem for MCs and MFAs. 
In~\cref{sec:related}, we discuss related work, and in~\cref{sec:conclusion} we conclude this paper.

{\bf{Notation}}. We write $\mathcal{D}$ for the distribution functor (with finite supports), and $\subdist$ and $\subdistcount$ for the subdistribution functor with finite supports and countable supports, respectively. 
We also write $\mfunc$ and $\mfunccount$ for the multiset functors with finite supports and countable supports, respectively. 
The constant function mapping every element to a fixed value $a$ is written as $\Delta_a$.
  We write the unit interval as $[0, 1]$,
and the set of booleans $\{\bot, \top\}$ as $\bools$.
We use $\pi_1$ and $\pi_2$ for the first and second projections, respectively.
\section{Preliminaries}
\label{sec:preliminaries}
We recall coalgebraic semantics for transition systems, which has been widely used in the literature (e.g.~\cite{HinoKH016,UrabeHH17,Kori0RK24,WatanabeJRH25}).
The semantics used in this paper is based on least fixed-point semantics, analogous to weakest precondition semantics~\cite{Hasuo15,AguirreKK22}, actively employed in studies of formal verification~\cite{HinoKH016,UrabeHH17,KoriUKSH22}.
Given a coalgebra $c\colon X\rightarrow FX$ of an endofunctor $F$ modelling a system, 
the semantics of $c$ is given by a least fixed point of a specific \emph{predicate transformer} of $c$ that is induced by a \emph{modality}. 

Throughout the paper, we consider coalgebras $c\colon X\rightarrow FX$ on the category of sets ($\sets$), thus $F$ is an endofunctor on the category of sets. A \emph{semantic structure} for $F$ is given by a pair $(\mathbf{\Omega}, \tau)$ such that (i) $\mathbf{\Omega}$ is an $\omega$-complete partially ordered set $(\Omega, \preceq)$ with the least element $\bot$; 
and (ii) $\tau$ is a function $\tau\colon F\Omega\rightarrow \Omega$. We call $\mathbf{\Omega}$ and $\tau$ a \emph{semantic domain} and a \emph{modality} (for $F$), respectively. 
\begin{definition}[predicate transformer]
    Given a coalgebra $c\colon X\rightarrow FX$ and a semantic structure $(\mathbf{\Omega}, \tau)$ for $F$, the \emph{predicate transformer} $\Phi_{c}$ of $c$ is a function $\Phi_c\colon \Omega^X \rightarrow \Omega^X$ given by $\Phi_c(u) \defeq \tau \circ F(u)\circ c$ for each $u\in \Omega^X$. 
    Additionally, we assume that $\Phi_{c}$ is $\omega$-continuous with respect to the pointwise order in $\Omega^X$. 
\end{definition}
The predicate transformer $\Phi_c$ is the composition of the coalgebra $c$ with  the predicate lifting $\tau \circ F(\_)$.
The latter corresponds bijectively  to the modality $\tau\colon F\Omega \to \Omega$, which is commonly used in coalgebraic modal logic (cf.~\cite{DBLP:journals/tcs/Schroder08}).

By the Kleene-fixed point theorem, the predicate transformer $\Phi_c$ has the least fixed point $\mu \Phi_c = \bigvee_{n\in \nat} \Phi^n(\bot)$. 
\begin{definition}[semantics]
    Given a coalgebra $c$ and a semantic structure $(\mathbf{\Omega}, \tau)$, 
    the \emph{semantics} is the least fixed point $\mu \Phi_c \in \Omega^X$ of the predicate transformer $\Phi_c$. 
\end{definition}

We recall the coalgebraic semantics of (labelled) Markov chains (MCs) and \emph{deterministic finite automaton (DFA)}.   
\begin{example}[semantics of MC] \label{eg:mc}
    We define a \emph{(labelled) MC} as a coalgebra $c\colon X\rightarrow\mcfunc{X}$. 
    The semantic structure $(\mathbf{\Omega}, \tau)$ is given by (i) $\mathbf{\Omega} \defeq (\traceMC, \preceq)$, where $\preceq$ is the pointwise order; and (ii) $\tau \colon \mcfunc{\traceMC} \rightarrow \traceMC$ is given by 
    \begin{align*}
        \tau(\sigma, a)(w) \defeq \begin{cases*}
            \sigma(\checkmark) &if  $w = a$,\\
            \sum_{\mu \in \subdistcount(A^+)}\sigma(\mu) \cdot \mu(w') &if   $w = a\cdot w'$ for some $w' \in A^+$,\\
            0 &otherwise.
        \end{cases*}
    \end{align*}
    The least fixed point of $\Phi_c$ gives 
    the reachability probability $\mu \Phi_c(x)(w)\in [0, 1]$ from each state $x\in X$ to the target $\checkmark$ with the trace $w\in A^{+}$. 
\end{example}

\begin{example}[semantics of DFA]
    We define a DFA\footnote{Following~\cite{WatanabeJRH25}, the acceptance condition is given by the output, as in Mealy machines.} as a coalgebra $d\colon Y\rightarrow \dfafunc{Y}$. 
    The semantic structure $(\mathbf{\Omega}, \tau)$ is given by (i) 
    $\mathbf{\Omega} \defeq(\traceDFA, \subseteq)$, where $\subseteq$ is the inclusion order; and (ii) $\tau \colon \dfafunc{\traceDFA} \rightarrow \traceDFA$ is given by 
    \begin{align*}
        \tau(\delta)\defeq \big\{a\in A \mid \pi_2 \big(\delta(a)\big) = \top\big\} \cup \big\{a\cdot w\in A^{+} \mid w \in \pi_1 \big(\delta(a)\big) \big\}.
    \end{align*}
    The least fixed point of the predicate transformer $\Phi_d$ yields
    the recognized language excluding the empty string, as $\mu \Phi_d(y)\in \traceDFA$ 
    for each state $y\in Y$. 
\end{example}

\section{Coalgebraic Product Construction}
\label{sec:coprodconst}
We introduce the coalgebraic product construction~\cite{WatanabeJRH25}, which is the foundation for our study of product constructions.
The coalgebraic product construction employs 
a natural transformation $\lambda$ called a \emph{distributive law}
to merge behaviours of two coalgebras; 
see, e.g.,~\cite{MacLane} for preliminaries on category theory.

\begin{definition}[distributive law, coalgebraic product]
    Let $F_S$, $F_R$, and $F_{S\otimes R}$ be endofunctors.
    A \emph{distributive law} $\lambda$ from $F_S$ and $F_R$ to $F_{S\otimes R}$ is a natural transformation 
    $\lambda\colon \times \circ (F_S \times F_R) \Rightarrow F_{S\otimes R}\circ \times$. 
    Given two coalgebras $c\colon X\rightarrow F_S X$ and $d\colon Y\rightarrow F_R Y$,
     the \emph{coalgebraic product} $c\otimes_{\lambda} d$ induced by the distributive law $\lambda$ is defined as the coalgebra: 
     \begin{displaymath}
     c\otimes_{\lambda} d\colon X\times Y\rightarrow F_{S\otimes R} (X\times Y),  \quad c\otimes_{\lambda} d \defeq \lambda_{X, Y}\circ (c\times d). 
     \end{displaymath}
\end{definition}

\begin{example}[product of MC and DFA]
    \label{ex:prodMCDFA}
    Define a distributive law $\lambda$ from $\mcfunc{(\_)}$ and $\dfafunc{(\_)}$ to $\mcfuncproduct{(\_)}$ by
    \begin{align*}
        \lambda_{X, Y}(z)(x, y)&\defeq \begin{cases}
            \sigma(x) &\text{ if } y = \pi_1 \big(\delta(a)\big),\\
            0 &\text{ otherwise,}
        \end{cases}\\ 
         \lambda_{X, Y}(z)(\checkmark) &\defeq \begin{cases}
            \sigma(\checkmark) &\text{ if } \top = \pi_2 \big(\delta(a)\big),\\
            0 &\text{ otherwise,}
        \end{cases}
    \end{align*}
    where $z = (\sigma, a, \delta)\in \mcfunc{X}\times\dfafunc{Y}$.
    Then the coalgebraic product $c \otimes_\lambda d$ of an MC $c\colon X\rightarrow\mcfuncproduct{X}$ and a DFA $d\colon Y\rightarrow \dfafunc{Y}$ is the standard product of the MC and the DFA, e.g.~\cite{baierKatoen, WatanabeJRH25}.
\end{example}

It is worth emphasizing that the coalgebraic product $c\otimes_{\lambda} d$ is itself a coalgebra.
This allows its semantics to be defined in the same manner as for its components, using a semantic structure $(\mathbf{\Omega}_{S\otimes R}, \tau_{S\otimes R})$.
\begin{example}[semantics of the product]
  \label{ex:semanticsProduct}
  Consider~\cref{ex:prodMCDFA}.
  We define the semantic structure $(\mathbf{\Omega}_{S\otimes  R}, \tau_{S\otimes R})$ as follows: (i) 
  $\mathbf{\Omega}_{S\otimes R} \defeq ([0, 1], \leq)$, where $\leq$ is the standard order; and (ii) $\tau_{S\otimes R}\colon \mcdfafunc{[0, 1]} \rightarrow [0, 1]$ is given by $\tau_{S\otimes R}(\sigma) \defeq \sigma(\checkmark) + \sum_{r\in [0, 1]}r \cdot \sigma(r)$. 
  Under this semantic structure, the semantics $\mu \Phi_{c \otimes_\lambda d}$ of the product gives, for each $(x, y)$, the reachability probability from $(x, y)$ to the target state $\checkmark$. 
\end{example}

\begin{remark}[final coalgebra semantics] 
  In \cref{ex:semanticsProduct}, one might wonder whether the semantics arises as a (unique) coalgebra morphism into a final coalgebra in the Kleisli category of the monad 
$\mcdfafunc{\_}$, thereby instantiating the generic trace semantics of Hasuo et al.~\cite{HasuoJS07}.
This, however, is not the case: the semantics is defined as the least fixed point of a predicate transformer, whereas the transformer may in general admit multiple fixed points.
Uniqueness of the coalgebra morphism can be ensured only under additional assumptions such as finite state spaces and almost-sure termination of coalgebras~\cite{baierKatoen}.
\end{remark}

We then move on to the \emph{correctness} of coalgebraic products.
This notion is defined with respect to an \emph{inference map} $q\colon \Omega_S\times \Omega_R \rightarrow \Omega_{S\otimes R}$, where $\Omega_S$ and $\Omega_R$ are the underlying sets of semantic domains $\mathbf{\Omega_S}$ and $\mathbf{\Omega_R}$  of $c$ and $d$, respectively.  

\begin{definition}[inference map, correctness]
    Let $c\colon X\rightarrow F_S X$ and $d\colon Y\rightarrow F_R Y$ be two coalgebras, $\lambda$ be a distributive law $\lambda_{X, Y}\colon (F_S X) \times (F_R Y)\rightarrow F_{S\otimes R} (X\times Y)$, and $(\mathbf{\Omega}_S, \tau_S)$,  $(\mathbf{\Omega}_R, \tau_R)$, and  $(\mathbf{\Omega}_{S\otimes R}, \tau_{S\otimes R})$ be semantic structures for $F_S, F_R$ and $F_{S\otimes R}$, respectively. 
    An \emph{inference map} is a function $q\colon \Omega_S\times\Omega_R \rightarrow \Omega_{S\otimes R} $ such that 
    \begin{enumerate}
      \item  $q(\bot, \bot) = \bot$; and
      \item $q$ is $\omega$-continuous, that is, $\bigvee_{l\in \nat} q(u_l,  v_l) = q\big(\bigvee_{m\in \nat}u_m,  \bigvee_{n\in \nat}v_n \big)$ for each $\omega$-chains $(u_m)_{m\in \nat}, (v_n)_{n\in \nat}$.    
    \end{enumerate}
    The coalgebraic product $c\otimes_{\lambda} d$ is said to be \emph{correct w.r.t.~$q$} if 
        $q \circ (\mu \Phi_c \times \mu \Phi_d) = \mu \Phi_{c\otimes_{\lambda} d}$ holds. 
\end{definition}

\begin{example}
    \label{ex:correctnessMCDFA}
    Consider~\cref{ex:prodMCDFA,ex:semanticsProduct}.
    The coalgebraic product $c \otimes_\lambda d$ is correct w.r.t.~$q\colon\traceMC\times \traceDFA\rightarrow [0, 1]$ defined by $q(\sigma, L) \defeq \sum_{w\in L}\sigma(w)$.
    This means that the semantics of the products---the reachability probabilities---coincides with the probability that traces of MCs are accepted by DFAs.  
\end{example}

A simple \emph{correctness criterion} that ensures the correctness of the coalgebraic product is proposed in~\cite{WatanabeJRH25}. 
This criterion is indeed general enough to capture a wide range of known product constructions in the literature~\cite{baierKatoen,AndovaHK03,BaierDDKK14,BaierAFK18}, including the one illustrated in~\cref{ex:correctnessMCDFA}.  
\begin{proposition}[correctness criterion~\cite{WatanabeJRH25}]
    \label{prop:correctCriterion}
    Assume the following data: 
    \begin{itemize}
        \item a distributive law $\lambda$ from $F_S$ and $F_R$ to $F_{S\otimes R}$, 
        \item semantic structures $(\mathbf{\Omega}_S, \tau_S)$,  $(\mathbf{\Omega}_R, \tau_R)$, and  $(\mathbf{\Omega}_{S\otimes R}, \tau_{S\otimes R})$ for $F_S$, $F_R$, and $F_{S\otimes R}$, respectively,
        \item an inference map $q\colon \Omega_S \times \Omega_R \to \Omega_{S \otimes R}$.
    \end{itemize}
    Then for any coalgebras $c\colon X\rightarrow F_SX$ and  $d\colon Y\rightarrow F_RY$,
    the coalgebraic product $c\otimes_{\lambda} d$ is correct w.r.t.~$q$ if the following equation holds: 
      \begin{align*}
        q\circ (\tau_S\times \tau_R) = \tau_{S\otimes R}\circ F_{S\otimes R}(q)\circ \lambda_{\Omega_S, \Omega_R}.
      \end{align*}
\end{proposition}
The correctness criterion requires that a semantic structure $(\mathbf{\Omega}_{S\otimes R}, \tau_{S\otimes R})$ for products be explicitly specified.
This requirement arises from the practical need for efficient computation of the semantics of products.
Model checking problems should ideally be solvable by mature techniques.
For instance, the semantics described in~\cref{ex:semanticsProduct} can be efficiently computed by solving linear equation systems, or applying value iterations (see~\cite{baierKatoen}). 
\begin{remark}[the choice of data]
  The coalgebraic product construction and its correctness criterion require several pieces of data to specify both the problem and its correctness.
These data are not necessarily 
canonical.
Ideally, one would like to determine 
these data
via universal properties---for instance, through final coalgebras.
We leave this direction for future work. 
\end{remark}

\section{Natural Transformations for Coalgebraic Product Constructions}
\label{sec:naturaltrans}
As demonstrated in~\cref{prop:correctCriterion}, 
the coalgebraic product construction requires a distributive law $\lambda$, semantic structures $(\Omega_{i}, \tau_i)$ for $i \in \{S, R, S \otimes R\}$, and an inference map $q$.
Before exploring these structures in concrete settings,
we begin by studying natural transformations between endofunctors on $\sets$ that arise from commutative monoids (see Def.~\ref{def:func_cmon}).
The results established in this section enable us to prove a no-go theorem and uniqueness of distributive laws for coalgebraic products in later sections.

\newcommand{\CMon}{\mathbf{CMon}}
We write $\CMon$ for the category of commutative monoids.
For a commutative monoid $A$,
the binary operator is denoted by $+_A$ and  the unit element by $0_A$;
when the context is clear, 
we omit these subscripts for simplicity.
For $n \in \mathbb{N}$ and $a \in A$, we define $n \cdot a$ as the sum of $n$ copies of $a$, with $0 \cdot a = 0$.
\begin{definition} \label{def:func_cmon}
  We define a \emph{functor} $F_{(\_)}\colon \CMon \to [\sets, \sets]$ as follows. 
  For each $A \in \CMon$, the \emph{functor} $F_A$ is defined by 
  \begin{align*}
    &F_A(X) \coloneqq \{h\colon X \to A \mid \supp(h) \text{ is finite} \}, 
    &&F_A(g)(f) \coloneqq \sum_{x \in g^{-1}(\_)} f(x), 
  \end{align*}
  for each $X \in \sets$, $g\colon X \to Y$, and $f\in F_A(X)$,   where $\supp(h) = \{x \in X \mid h(x) \neq 0\}$.
  For each $i\colon A \to B$ in $\CMon$, the \emph{natural transformation} $F_i$ is defined by $(F_i)_X(f) \coloneqq i \circ f$ for each $X\in \sets$ and $f\in F_A(X)$. 
\end{definition}
For a fixed monoid $A$, the functor $F_A$ has appeared in the literature on coalgebraic transition systems labelled by monoids~\cite{DBLP:journals/entcs/GummS01}.

\begin{example} \label{eg:cmon_functor}
  \begin{enumerate}
    \item 
    The functor $F_{(\mathbb{N}, +, 0)}$ is the multiset functor $\mfunc$.
    \item 
    The functor $F_{(\bools, \lor, \bot)}$ is
    the covariant finite powerset functor $\powfunc$
    where  $\lor$ is the logical OR.
    \item 
    The functors $F_{(\nnegreal, +, 0)}$ and $F_{(\nnegreal, \times, 1)}$ 
    both map a set $X$ to a set of non-negative real-valued functions on $X$ with finite support, while they map functions in a different way.
    We write $\erealfunc$ and $\erealfunctimes$ for these functors $F_{(\nnegreal, +, 0)}$ and $F_{(\nnegreal, \times, 1)}$, respectively.
  \end{enumerate}
\end{example}

For $n \in \mathbb{N}$,
we use $\setn{n}$ to represent the set $\{1, \cdots, n\}$,
with the convention that $\setn{0}$ refers to the empty set.
\begin{proposition} \label{prop:natural_trans_2}
  Any natural transformation $\lambda\colon F_A \Rightarrow F_B$
  is uniquely determined by its component at $\setn{2}$,
  that is,
  for each $\lambda, \lambda'\colon F_A \Rightarrow F_B$,
  $\lambda_\setn{2} = \lambda'_\setn{2}$ if and only if $\lambda = \lambda'$.
  \qed
\end{proposition}
See \ifarxiv
\cref{subsec:proof_prop_natural_trans_2}
\else
\cite[Appendix A.1]{KWfull}     
\fi
for the proof.
Note that 
  each collection
  $\Nat{F_A}{F_B}$ is a set.

In this section, we aim to precisely characterise this set 
$\Nat{F_A}{F_B}$,
providing explicit constructions for such natural transformations.

\begin{remark}
    By analogy with Yoneda's lemma, 
    it makes sense to speculate that $\Nat{F_A}{F_B}$ can be characterised by monoid morphisms from $A$ to $B$.
    However, this is not the case in general: while the functor $F_{(\_)}$ is faithful
    \ifarxiv
    (as shown in~\cref{ap:faithful}),
    \else
    (as shown in~\cite[Appendix A.2]{KWfull}),
    \fi 
    it fails to be full.
    To see this,
    consider $A=(\nat, +, 0)$ and $B=(\nnegreal, +, 0)$.
    It will be shown in \Cref{eg:mfunc_natural_trans} that
    each natural transformation $\lambda \in \Nat{F_{(\nat, +, 0)}}{F_{(\nnegreal, +, 0)}}$ is of the form
    $\lambda_X(f)(x) = f(x) \cdot b\big(\sum_{x \in X}f(x)\big)$ for some $b\colon \nat \to \nnegreal$.
    Assume, for contradiction, that there exists $i\colon \nat \to \nnegreal$ such that $\lambda = F_i$.
    Since $i$ is uniquely determined by the component $\lambda_1$, we have $\lambda_1(\Delta_n)(1) = i(n) = n \cdot b(n)$.
    Since $F_i(f)(x) = i(f(x)) = f(x) \cdot b(f(x))$, $F_i(f)(x)$ is in general different from $\lambda_X(f)(x)$.
\end{remark}
After examining some fundamental properties of these natural transformations,
we proceed to explore their explicit forms in certain cases.

\begin{lemma} \label{lem:monoid_lambda}
  Let $A$ and $B$ be commutative monoids,
  and let $\lambda\colon F_A \Rightarrow F_B$ be a natural transformation.
  Then, 
  for each $f \in F_A(X)$ and $x, x' \in X$, the following statements hold.
  \begin{enumerate}
    \item \label{item:f_equal} $f(x) = f(x')$ implies $\lambda_X(f)(x) = \lambda_X(f)(x')$.
    \item \label{item:idempotent} 
    For each $n \in \mathbb{N}$,
    $n \cdot f(x) = f(x)$ implies $n \cdot \lambda_X(f)(x) = \lambda_X(f)(x)$.
  \end{enumerate}
\end{lemma}
\begin{proofs}
  \eqref{item:f_equal} 
  This is easy to prove by taking a function $g\colon X\rightarrow X$ that swaps $x$ and $x'$ and by the naturality of $\lambda$. 

  \eqref{item:idempotent}
  The cases $n = 0$ and $n=1$ are easy to prove. 
  We sketch the proof for the case $n \geq 2$. 
  Let $\kappa_2$ be the second coprojection of a binary coproduct. 
  We define $f'\in F_A(X+(\setn{2n-2}))$ as $[f, \Delta_{f(x)}]$ 
  and show that 
  \begin{align*}
    \lambda_X(f)(x) &= (2n-1)\cdot \lambda_{X + (\setn{2n-2})}(f')(\kappa_2(n)),\text{ and }\\
     \lambda_{X + (\setn{2n-2})}(f')(\kappa_2(n)) &= n\cdot \lambda_{X + (\setn{2n-2})}(f')(\kappa_2(n)).
  \end{align*}
  Then, we can see that $\lambda_X(f)(x) = \lambda_{X + (\setn{2n-2})}(f')(\kappa_2(n))$
  because 
  \begin{align*}
    \lambda_X(f)(x) &= (2n-1)\cdot \lambda_{X + (\setn{2n-2})}(f')(\kappa_2(n))\\
    &= n\cdot \lambda_{X + (\setn{2n-2})}(f')(\kappa_2(n)) + (n-1)\cdot \lambda_{X + (\setn{2n-2})}(f')(\kappa_2(n)) \\
    &= \lambda_{X + (\setn{2n-2})}(f')(\kappa_2(n)) + (n-1)\cdot \lambda_{X + (\setn{2n-2})}(f')(\kappa_2(n)) \\
    &= \lambda_{X + (\setn{2n-2})}(f')(\kappa_2(n)).
  \end{align*}
This concludes that $\lambda_X(f)(x) = n\cdot \lambda_X(f)(x)$. 
\ifarxiv 
See~\cref{subsec:proof_lem_monoid_lambda} for the full proof. 
\else   
See~\cite[Appendix A.3]{KWfull} for the full proof. 
\fi
\end{proofs}

The second statement places a restriction on the possible values of $\lambda_X(f)(x)$ when $n \cdot f(x) = f(x)$.
In particular, setting $n=0$ 
implies that 
$\lambda_X(f)(x) = 0$ if $f(x) = 0$.

The following result tells how 
$\lambda$ respects the $n$-fold operation.
To formalize this,
we introduce a preorder on a commutative monoid $A$ defined by
$a \leq a'$ if and only if there exists $a'' \in A$ such that $a + a'' = a'$.
\begin{lemma} \label{lem:lambda_m_n}
  Let $A$ and $B$ be commutative monoids,
  and $\lambda\colon F_A \Rightarrow F_B$ be a natural transformation.
  For each $f \in F_A(X)$, $x \in X$, $a \in A$, $n, m \in \mathbb{N}$, 
  and
  $\triangleright_1, \triangleright_2 \in \{=, \geq\}$,
  we assume that 
  \begin{center}
  $f(x) \triangleright_1 n \cdot a$,\quad and\quad  $\sum_{x' \in X \setminus \{x\}}f(x') \triangleright_2 m \cdot a$.
  \end{center}
  Then there are $a_1, a_2 \in A$ such that $f(x) = n \cdot a + a_1$ and $\sum_{x' \in X \setminus \{x\}}f(x') = m \cdot a + a_2$ by definition of $\triangleright_1$ and $\triangleright_2$.
  For these $a_1$ and $a_2$, the following relations hold:
  \begin{center}
  $\lambda_X(f)(x) \triangleright_1 n \cdot \lambda_{\setn{n+m+2}}(f')(1)$, \quad $\sum_{x' \in X \setminus \{x\}} \lambda_X(f)(x') \triangleright_2 m \cdot \lambda_{\setn{n+m+2}}(f')(1)$,
  \end{center}
  where $f'\in F_A(\setn{n+m+2})$ is defined by $f'(n+m+1) \defeq a_1$, $f'(n+m+2) \defeq a_2$, and $f'(i) \defeq a$ for each $i \in \setn{n+m}$. 
\end{lemma}
\begin{proof}
  Define $g_1\colon \setn{n+m+2} \to \setn{2}$ by $g_1(i) = 1$ if $i \in \setn{n} \cup \{n+m+1\}$ and $2$ otherwise,
  and $g_2\colon X \to \setn{2}$ by $g_2(x') = 1$ if $x' = x$ and $2$ otherwise.
  Then the following equalities hold.
  \begin{align*}
    &\lambda_X(f)(x)
    = \lambda_\setn{2}(F_A(g_2)(f))(1) &&\text{by naturality for }g_2, \\
    &= \lambda_\setn{2}(F_A(g_1)(f'))(1) &&\text{since }F_A(g_2)(f) = F_A(g_1)(f'), \\
    &= n \cdot \lambda_{\setn{n+m+2}}(f')(1) + \lambda_{\setn{n+m+2}}(f')(n+m+1) &&\text{by naturality and \Cref{lem:monoid_lambda}.\ref{item:f_equal}}.
  \end{align*}
  If $\triangleright_1$ is equal to $=$, then 
  we can choose $a_1 = 0$, and then
  $\lambda_X(f)(x) = n \cdot \lambda_{\setn{n+m+2}}(f')(1)$ since
  $\lambda_{\setn{n+m+2}}(f')(n+m+1) = 0$, which follows from \cref{lem:monoid_lambda}.\ref{item:idempotent} with $n=0$.
  If $\triangleright_1$ is equal to $\geq$, then $\lambda_X(f)(x) \geq n \cdot \lambda_{\setn{n+m+2}}(f')(1)$ clearly holds.
  One can also prove 
    $\sum_{x' \in X \setminus \{x\}} \lambda_X(f)(x') \triangleright_2 m \cdot \lambda_{\setn{n+m+2}}(f')(1)$
  by considering $\lambda_\setn{2}(F_A(g_2)(f))(2)$ instead of $\lambda_\setn{2}(F_A(g_2)(f))(1)$.
\end{proof}

\subsection{For Natural Transformations from $F_A$ Where $A$ is Singly Generated} \label{sec:nt_single}

Let us focus on a commutative monoid $A$ generated by a single element $a \in A$.
In this setting, 
each $a' \in A$ can be written by the form $l \cdot a$ for some $l \in \mathbb{N}$.
We define $N\colon A \to \mathbb{N}$ that assigns to each $a' \in A$ the least natural number $l$ such that $l \cdot a = a'$.

For a singly generated commutative monoid $A$,
\cref{lem:lambda_m_n} provides an insight into 
the explicit form of $\lambda$. 
Informally, given a natural transformation $\lambda\colon F_A \Rightarrow F_B$ and $f\in F_A(X)$, the value $\lambda_X(f)(x)$ can be expressed  as 
$\lambda_X(f)(x) = N(f(x))\cdot O(f)$ for any $x\in X$, 
where $O\colon F_A(X) \to B$ is a certain function. 
Moreover, the function $O$ depends only on the sum of values of $f$.
This follows from the fact that
the naturality of $\lambda$ imposes conditions only on functions whose total sums are equal
because for any $g\in X\rightarrow Y$,
$F_A(g)$ preserves the sum of function values, i.e.~$\sum_{x \in X} f(x) = \sum_{y \in Y}(F_A(g)(f))(y)$.

\begin{theorem} \label{thm:monoid_lambda}
  Let $A$ be a commutative monoid 
  generated by a single element $a \in A$,
  and let $B$ be a commutative monoid.
 In the following two cases, the set $\Nat{F_A}{F_B}$ of natural transformations 
 can be explicitly characterised:
  \begin{enumerate}
    \item \label{item:case_n} Case 1: for each $n, m \in \mathbb{N}$, $n \neq m$ implies $n \cdot a \neq m \cdot a$.

    There exists an isomorphism 
    $\lambda^{(\_)}\colon \{b \in B^{\mathbb{N}} \mid b(0) = 0 \} \to \Nat{F_A}{F_B}$:
  \begin{displaymath}
    \lambda^b_X(f)(x) \coloneqq
    N(f(x)) \cdot b\big(\sum_{x \in X}N(f(x))\big).
  \end{displaymath}

    \item \label{item:case_cycle} Case 2: 
    $a \neq 0$ and there is $n \in \mathbb{N}_{> 1}$ such that $n \cdot a = a$ and
    $(n-1) \cdot a \neq 0$.

    Let $n$ be the least natural number such that $n > 1$ and $n \cdot a = a$.
   There exists an isomorphism $\lambda^{(\_)}\colon \{c \in B \mid n \cdot c = c\}^{\{0, \cdots, n-2\}} \to \Nat{F_A}{F_B}$ given by
  \begin{displaymath}
    \lambda^b_X(f)(x) \coloneqq
    N(f(x)) \cdot b\Big(\big[\sum_{x \in X}N(f(x))\big]\Big),
  \end{displaymath}
  where $[l]$ is the remainder of $l$ modulo $n-1$.
  \end{enumerate}
\end{theorem}

\begin{proofs}
  1) For each $n\in \mathbb{N}_{\geq 1}$, we use the function $f'_n \in F_A(\setn{n+2})$ given by $f'_n(i) = a$ for each $i \in \setn{n}$ and $f'_n(n+1) = f'_n(n+2) = 0$.
  We define the inverse $(\lambda^{(\_)})^{-1}\colon \Nat{F_A}{F_B}\rightarrow \{b \in B^\nat \mid b(0) = 0\}$ as follows: 
  Given a natural transformation $\lambda$, the inverse $b^{\lambda}\colon \mathbb{N} \to B$ is given by $b(0) = 0$ and $b(n) = \lambda_{\setn{n+2}}(f'_n)(1)$ for each $n \in \mathbb{N}_{\geq 1}$. 
  We can show that the mapping $b^{(\_)}$ is indeed the inverse by~\cref{lem:lambda_m_n}. 
  
  2) 
  The construction of the inverse is more involved. 
  We use the same function $f'_n$ for each $n\in \mathbb{N}_{\geq 1}$. 
  We define the inverse $(\lambda^{(\_)})^{-1}\colon \Nat{F_A}{F_B}\rightarrow  \{c \in B \mid n \cdot c = c\}$ as follows:
  Given a natural transformation $\lambda$, the inverse $b^{\lambda}\colon \{0, \cdots, n-2\} \to \{c \in B \mid n \cdot c = c\}$ is given by $b^\lambda(m) \defeq d^\lambda([m-1]+1)$, 
  where the function $d^\lambda\colon \mathbb{N}_{\geq 1} \to \{c \in B \mid n \cdot c = c\}$ is defined by $d^{\lambda}(m) \defeq \lambda_{\setn{m+2}}(f'_m)(1)$. 
  We can show that this mapping is the inverse by~\cref{lem:monoid_lambda}.\ref{item:idempotent} and~\cref{lem:lambda_m_n}. 
  
  \ifarxiv
  See~\cref{appendix:proof_monoid_lambda} for the full proof. 
  \else   
  See~\cite[Appendix A.4]{KWfull} for the full proof. 
  \fi
\end{proofs}

Let us instantiate this theorem with concrete examples.
\begin{example}[$\mfunc \Rightarrow F_B$] \label{eg:mfunc_natural_trans}
Consider the multiset functor $\mfunc = F_{(\mathbb{N}, +, 0)}$ (see \cref{eg:cmon_functor}).
The monoid $(\mathbb{N}, +, 0)$ is generated by $1$,
and
this falls under Case 1 of \cref{thm:monoid_lambda}.
Applying the theorem, we derive the following results:
  \begin{enumerate}
    \item
    There exists an isomorphism
    $\lambda\colon \{b \in\bools^{\mathbb{N}} \mid b(0) = \bot \} \to \Nat{\mfunc}{\powfunc}$ defined by 
      $\lambda_X^b(f)\coloneqq 
      \{x \in X \mid f(x) > 0 \text{ and } b(\sum_{x \in X} f(x))\}$.
    \item
    There exists an isomorphism
    $\lambda\colon \{b \in (\nnegreal)^{\mathbb{N}} \mid b(0) = 0\} \to \Nat{\mfunc}{\erealfunc}$ defined by 
      $\lambda_X^b(f)(x) \coloneqq 
      f(x) \cdot b(\sum_{x \in X} f(x))$.
    \item
    There exists an isomorphism
    $\lambda\colon \{b\in \nnegreal^{\mathbb{N}} \mid b(0) = 1\} \to \Nat{\mfunc}{\erealfunctimes}$ defined by 
      $\lambda_X^b(f)(x) \coloneqq 
      \big(b\big(\sum_{x \in X} f(x)\big)\big)^{f(x)}$.
  \end{enumerate}
\end{example}

\begin{example}[$\powfunc \Rightarrow F_B$] \label{eg:powfunc_natural_trans}
Consider the finite powerset functor $\powfunc = F_{(\bools, \lor, \bot)}$ (see \cref{eg:cmon_functor}).
This falls under Case 2 of \cref{thm:monoid_lambda}
since $(\bools, \lor, \bot)$ is generated by $\top$ that is idempotent.
By~\cref{thm:monoid_lambda}, natural transformations are constrained by the idempotent property of $\top$, and
the isomorphism takes the form:
$\lambda^{(\_)}\colon \{c \in B \mid c^2 = c\} \to \Nat{\powfunc}{F_B}$.
We then derive the following results:
  \begin{enumerate}
    \item There exists a unique natural transformation $\lambda\colon \powfunc \Rightarrow \erealfunc$
    (respectively, $\lambda\colon \powfunc \Rightarrow \mfunc$),
    which is
    given by $\lambda_X(S)(x) = 0$.
    \item There exists only two natural transformations $\lambda\colon \powfunc \Rightarrow \erealfunctimes$. 
    One is given by $\lambda_X(S)(x) \defeq 1$; and the other is by $\lambda_X(S)(x) \defeq 0$ if $x \in S$, and $1$ otherwise.
  \qed
  \end{enumerate}
\end{example}

The following lemma allows us to extend our results 
about natural transformations between $F_A$ and $F_B$ to their respective subfunctors.
\ifarxiv
See~\cref{ap:proof_subfunctor} for its proof.
\else   
See~\cite[Appendix A.5]{KWfull} for its proof.
\fi
\begin{lemma} \label{lem:subfunctor}
  Let $F$ and $G$ be subfunctors of $F_A$ and $F_B$, respectively,
  meaning that there exist $\iota_A\colon F \Rightarrow F_A$ and $\iota_B\colon G \Rightarrow F_B$ given by inclusions.
  Assume that for all $g\colon X \to Y$ in $\sets$ and $f \in F_A(X)$,
  $F_A(g)(f) \in F(Y)$ implies $f \in F(X)$.
  Then for any 
  natural transformation $\lambda\colon F \Rightarrow G$,
  there exists a natural transformation 
  $\lambda' \colon F_A \Rightarrow F_B$
  such that
  $\iota_B \circ \lambda = \lambda' \circ \iota_A$.
  \qed
\end{lemma}

\begin{proposition} \label{prop:subfunc_pm}
  \begin{enumerate}
    \item \label{item:pd} No natural transformation $\powfunc \Rightarrow \mathcal{D}$ exists.
    \item There exists a unique natural transformation $\lambda\colon \powfunc \Rightarrow \subdist$.
    \item No natural transformation $\mfunc \Rightarrow \mathcal{D}$ exists.
    \item
    There is an isomorphism $\lambda^{(\_)}\colon [0, 1]^{\mathbb{N}_{\geq 1}} \to \Nat{\mfunc}{\subdist}$ defined by
    \begin{displaymath}
      \lambda^b_X(f)(x) = 
      \begin{cases}
        \frac{f(x)}{\sum_{x \in X}f(x)} \cdot b\big(\sum_{x \in X}f(x)\big) &\text{if }\sum_{x \in X}f(x) > 0, \\
        0 &\text{if }\sum_{x \in X}f(x) = 0.
      \end{cases}
    \end{displaymath}
  \end{enumerate}
\end{proposition}
We provide a proof of statement~\ref{item:pd} below.
\ifarxiv 
The detailed proofs of statements 2, 3, and 4 are given in~\cref{ap:proof_subfunc_pm}. 
\else  
The detailed proofs of statements 2, 3, and 4 are given in~\cite[Appendix A.6]{KWfull}. 
\fi 
\begin{proof}[Proof of~\ref{item:pd}]
  Suppose that a natural transformation $\lambda\colon \powfunc \Rightarrow \mathcal{D}$ exists.
  Since $\mathcal{D}$ is a subfunctor of $\erealfunc$, 
  \cref{lem:subfunctor} and \cref{thm:monoid_lambda}.\ref{item:case_cycle} imply that
  $\lambda$ should be 
  the constant zero transformation.
  It may have an image that is not contained in $\mathcal{D}$. 
  Consequently,
  no natural transformation $\powfunc \Rightarrow \mathcal{D}$ exists.
\end{proof}

\subsection{For Natural Transformations from $\erealfunc$}
We further investigate natural transformations $\lambda$ from $\erealfunc = F_{(\nnegreal, +, 0)}$ (see \cref{eg:cmon_functor}) to some $F_B$.
\cref{thm:monoid_lambda} is not applicable here, 
as $(\nnegreal, +, 0)$ is not generated by a single element.
However,
we can still determine the explicit form of $\lambda$
by utilizing the property that
the value $\lambda_X(f)(x)$ is constrained by
the sum $\sum_{x' \in X}\lambda_X(f)(x')$ and
the ratio of $f(x)$ to the total sum $\sum_{x \in X}f(x)$.
\begin{lemma} \label{lem:erealfunc_b}
  Let $B$ be a commutative monoid, 
  $\lambda\colon \erealfunc \Rightarrow F_B$
  be a natural transformation.
  For any $f \in \erealfunc(X)$, $x \in X$,
  and $m, n \in \mathbb{N}$ with $n < m$ such that  $\frac{n}{m} \cdot \sum_{x' \in X}f(x') \leq f(x) \leq \frac{n+1}{m}\cdot \sum_{x' \in X}f(x')$, 
  the following statements hold.
  \begin{enumerate}
    \item \label{item:b_lambda_1} There is $b \in B$ such that
$(m-n) \cdot b = \sum_{x' \in X \setminus \{x\}}\lambda_X(f)(x')$ and
$n \cdot b \leq \lambda_X(f)(x)$,
and
\item There is $b \in B$ such that
$(n+1) \cdot b = \lambda_X(f)(x)$
and
$(m-n-1) \cdot b \leq \sum_{x' \in X \setminus \{x\}}\lambda_X(f)(x')$.
  \end{enumerate}
  \end{lemma}
\begin{proof}
  Let $r \coloneqq \sum_{x' \in X}f(x')$.
  Note that 
$\frac{n}{m} \cdot r \leq f(x) \leq \frac{n+1}{m}\cdot r$
is equivalent to 
$\frac{n}{m-n} \cdot (r-f(x)) \leq f(x)$ and $\frac{m-n-1}{n+1}\cdot f(x) \leq r-f(x)$.
\cref{lem:lambda_m_n} for $(f, x, \frac{r-f(x)}{m-n}, n, m-n, \geq, =)$
implies that
$n \cdot \lambda_{\setn{m+2}}(f')(1) \leq \lambda_X(f)(x)$
and
$(m-n) \cdot \lambda_{\setn{m+2}}(f')(1) = \sum_{x' \in X \setminus \{x\}}\lambda_X(f)(x')$
where $f'$ is defined in \cref{lem:lambda_m_n}.
Similarly, 
\cref{lem:lambda_m_n} for $(f, x, \frac{f(x)}{n+1}, n+1, m-n-1, =, \geq)$
implies the second statement. 
\end{proof}

When $F_B = \powfunc$,
\cref{lem:erealfunc_b} allows us to conclude that
for any $f \in \erealfunc(X)$ and $x \in X$,
$\lambda_X(f)(x) = \bigvee_{x' \in X}\lambda_X(f)(x')$ if $f(x) > 0$.
This leads to the following result.
\ifarxiv
See \cref{ap:proof_ereal_powfunc} for the proof.
\else      
See~\cite[Appendix A.7]{KWfull} for the proof.
\fi 
\begin{proposition} \label{prop:ereal_powfunc}
  There is an isomorphism $\lambda^{(\_)}\colon \{b \in\bools^{\nnegreal} \mid b(0) = \bot \} \to \Nat{\erealfunc}{\powfunc}$ defined by
    $\lambda^b_X(f) = 
      \{x \in X \mid f(x) > 0 \text{ and }b\big(\sum_{x \in X}f(x)\big)\}$.\qed
\end{proposition}
  
Similarly, by \cref{lem:erealfunc_b} when $F_B = \erealfunc$,
we can prove
that 
the ratio of $\lambda_X(f)(x)$ to $\sum_{x' \in X}\lambda_X(f)(x')$ is equal to the ratio of $f(x)$ to $\sum_{x' \in X}f(x')$.
This allows us to establish the following result.
\begin{proposition} \label{prop:ereal_ereal}
    There is an isomorphism $\lambda^{(\_)}\colon \{b \in (\nnegreal)^{\nnegreal} \mid b(0) = 0\} \to \Nat{\erealfunc}{\erealfunc}$ defined by
    $\lambda^b_X(f)(x) = 
    f(x) \cdot b\big(\sum_{x \in X}f(x)\big)$.
\end{proposition}
\begin{proof}
  The inverse $b^{(\_)}$ of $\lambda^{(\_)}$ is defined as follows: given  $\lambda\colon \erealfunc \Rightarrow \erealfunc$, the  function $b^{\lambda}\colon \nnegreal \to \nnegreal$ is defined by $b(0) \coloneqq 0$ and $b^{\lambda}(r) \coloneqq \frac{1}{r} \cdot \lambda_\setn{1}(\Delta_{r})(1)$. 
  We first prove that the mapping $\lambda^{(\_)}$ is the left inverse of $b^{(\_)}$. 
  Let $\lambda\colon \erealfunc \Rightarrow \erealfunc$, $f \in \erealfunc(X)$, and $x \in X$.
  If $f(x) = 0$, then clearly $\lambda_X(f)(x) = 0 = \lambda_X^{b^{\lambda}}(f)(x)$ by \cref{lem:monoid_lambda}.\ref{item:idempotent} with $n=0$.
  Otherwise (i.e.~$f(x) > 0$),  
  define $r \defeq \sum_{x' \in X}f(x')$.
  By \cref{lem:erealfunc_b}, for any natural numbers $m$ and $n$ such that $n < m$ and  $\frac{n}{m} \leq \frac{f(x)}{r} \leq \frac{n+1}{m}$, we obtain the inequalities:
  \begin{displaymath}
    \frac{n}{m} \cdot \bigg(\sum_{x' \in X}\lambda_X(f)(x')\bigg) \ \leq \ \lambda_X(f)(x)
    \ \leq \ \frac{n+1}{m} \cdot \bigg(\sum_{x' \in X}\lambda_X(f)(x')\bigg).
  \end{displaymath}
  Taking limits as $m \to \infty$,
these inequalities yield $\frac{f(x)}{r} \cdot \sum_{x' \in X}\lambda_X(f)(x') = \lambda_X(f)(x)$. 
Furthermore, by the naturality of $\lambda$ for $!_X$, we have
$\sum_{x' \in X}\lambda_X(f)(x') = \lambda_\setn{1}(\Delta_{r})(1)$, implying the equation $\lambda_X(f)(x) = \frac{f(x)}{r} \cdot \lambda_\setn{1}(\Delta_{r})(1) = \lambda_X^{b^{\lambda}}(f)(x)$. 
It remains to show that $\lambda^{(\_)}$ is the right inverse, which is easy to check.
\end{proof}

This result, together with \cref{lem:subfunctor}, allows us to analyze 
natural transformations for important cases involving the (sub)distribution functor and the multiset functor.
\ifarxiv
See~\cref{ap:proof_subdist_ereal} for the omitted proof.
\else   
See~\cite[Appendix A.8]{KWfull} for the omitted proof.   
\fi 
\begin{corollary} \label{cor:subdist_ereal}
  \begin{enumerate}
    \item \label{item:subdist_1} There is an isomorphism $\lambda^{(\_)}\colon \{b \in (\nnegreal)^{[0, 1]} \mid b(0) = 0 \} \to \Nat{\subdist}{\erealfunc}$ defined by
    $\lambda^b_X(f)(x) = 
    f(x) \cdot b(\sum_{x \in X}f(x))$.
  \item \label{item:subdist_2}
    There is an isomorphism $\lambda^{(\_)}\colon \nnegreal \to \Nat{\mathcal{D}}{\erealfunc}$ defined by
    $\lambda^b_X(f)(x) = 
    f(x) \cdot b$.
  \item \label{item:subdist_3}
    There is an isomorphism $\lambda^{(\_)}\colon [0, 1]^{(0, 1]} \to \Nat{\subdist}{\subdist}$ defined by
    \begin{displaymath}
      \pushQED{\qed}
      \lambda^b_X(f)(x) = 
      \begin{cases}
        \frac{f(x)}{\sum_{x \in X}f(x)} \cdot b\big(\sum_{x \in X}f(x)\big) &\text{if }\sum_{x \in X}f(x) > 0, \\
        0 &\text{if }\sum_{x \in X}f(x) = 0.
      \end{cases}
      \qedhere
      \popQED
    \end{displaymath}
  \end{enumerate}
\end{corollary}

\begin{corollary}
  There exists a unique natural transformation $\lambda\colon \erealfunc \to \mfunc$ given by $\lambda_X(f)(x) = 0$.
\end{corollary}
\begin{proof}
  Let $\lambda\colon \erealfunc \Rightarrow \mfunc$ be a natural transformation.
  By \cref{lem:subfunctor} and \cref{prop:ereal_ereal}, there exists a unique $b'\colon \nnegreal \to \nnegreal$ such that $b'(0) = 0$ and $\lambda_X(f)(x) = f(x) \cdot b'\big(\sum_{x \in X}f(x)\big)$ for each $f \in \erealfunc(X)$ and $x \in X$. 
  Now, assume for contradiction that
  there exists $r > 0$ such that $b'(r) > 0$.
  We can construct a function $f\in \erealfunc(\setn{2})$ such that
  $\sum_{x \in \setn{2}}f(x) = r$, $f(2) \neq 0$, and
  $\frac{f(1)}{f(2)}$ is an irrational number.
  Since $\lambda_{\setn{2}}(f)(1)$ and $\lambda_{\setn{2}}(f)(2)$ must be natural numbers, it follows that their ratio must be rational: $\frac{\lambda_{\setn{2}}(f)(1)}{\lambda_{\setn{2}}(f)(2)} = \frac{f(1) \cdot b'(r)}{f(2) \cdot b'(r)} = \frac{f(1)}{f(2)}$. This is a contradiction.
  Therefore, $b'(r)$ must be $0$ for each $r \in \mathbb{R}_{> 0}$.
\end{proof}

\section{Markov Chains and Non-deterministic Finite Automata} 
\label{sec:mcnfas}

To handle a specific model checking problem using the coalgebraic product construction,
it is necessary to specify the structural components described in \cref{prop:correctCriterion}:
a distributive law,
semantic structures for both components and their composite,
and an inference map.
  Assume that the semantic structures for the individual components and the composite system, together with the appropriate inference map, are already determined.
This assumption is reasonable: the composite system is expected to be well studied and readily solvable by existing methods, and the model checking problem typically determines the relevant inference map.
For the coalgebraic product construction to be valid, we must identify a distributive law that satisfies the correctness criterion with the above data.
In this section, we prove that 
no such distributive law exists
for coalgebraic product constructions involving MCs and NFAs.

\begin{definition}[semantic structure of NFAs]
    \label{def:nfas}
    An \emph{NFA} is a coalgebra $d\colon Y\rightarrow \nfafunc{Y}$, where the underlying set $Y$ is finite.
    The \emph{semantic structure $(\mathbf{\Omega}_{R}, \tau_R)$ of NFAs} is defined by (i) $\mathbf{\Omega}_{R} \defeq (\traceNFA, \subseteq)$; 
    and (ii) $\tau_R\colon \nfafunc{\traceNFA}\rightarrow \traceNFA$ is given by 
    \begin{align*}
        \tau_R(\delta) \defeq 
        \{a \in A \mid \checkmark \in \delta(a)\} \cup \{a \cdot w\mid S \in \delta(a), w \in S\}.
    \end{align*}
\end{definition}
The semantics of NFAs is their recognized languages.
We write $\exreal\defeq \{r \in \real\mid r \geq 0\}\cup \{+\infty\}$ with the 
 convention $\infty \cdot 0 = 0 \cdot \infty \coloneqq 0$.

\begin{definition}[semantic structure of the product] \label{def:sem_product}
    The \emph{semantic structure} $(\mathbf{\Omega}_{S\otimes R}, \tau_{S\otimes R})$ is defined by 
    (i) $\mathbf{\Omega}_{S\otimes R} \defeq (\exreal, \leq)$, where $\leq$ is the standard order; and (ii) $\tau_{S\otimes R}\colon \mcnfafunc{\exreal}\rightarrow \exreal$ is given by 
    \begin{align*}
        \tau_{S\otimes R}(\sigma) \defeq \sigma(\checkmark) + \sum_{r\in \exreal} r\cdot \sigma(r). 
    \end{align*}
\end{definition}

We now present our first main result, a no-go theorem for coalgebraic product MCs for MCs and NFAs.
The proof relies on the characterisation of natural transformations from $\powfunc$ to $\erealfunc$ presented in \cref{sec:nt_single}.
The characterisation imposes strict constraints on possible forms of distributive laws,
and we demonstrate that no such distributive law satisfies the correctness criterion.
\begin{theorem}[no-go theorem under correctness criterion]
    \label{thm:no-go-mc-nfa}
    Consider semantic structures defined in \cref{eg:mc}, \cref{def:nfas}, and \cref{def:sem_product}.
    There is no distributive law $\lambda$ from $\subdist(\_) \times A$ and $\nfafunc{(\_)}$ to $\mcnfafunc{(\_)}$
    such that for each MC $c$ and NFA $d$,
    the coalgebraic product $c\otimes_{\lambda} d$ satisfies the correctness criterion with the inference map $q\colon \traceMC\times \traceNFA \rightarrow \exreal$ defined by 
    \begin{align*}
    q(\sigma, L) \defeq \sum_{w\in L} \sigma(w), &&\text{for each $\sigma \in \traceMC$ and $L \in \traceNFA$.} 
    \end{align*}
\end{theorem}
\begin{proof}
  Suppose that there exists a distributive law $\lambda$ from $\subdist(\_) \times A$ and $\nfafunc{(\_)}$ to $\mcnfafunc{(\_)}$
  satisfying the correctness criterion.
  Since there is a natural isomorphism $\mcnfafunc{(\_)} \Rightarrow \erealfunc(\_) \times \nonnegreal$,
  there is a bijective correspondence between
  $\lambda$ and a pair of natural transformations:
  \begin{itemize}
    \item $\{\lambda'_{X, Y}\colon \subdist(X) \times A \times \nfafunc{Y} \to \nnegreal \}_{X, Y}$, and
    \item $\{\lambda''_{X, Y}\colon \subdist(X) \times A \times \nfafunc{Y} \to \erealfunc(X \times Y)\}_{X, Y}$.
  \end{itemize}
  We now analyze the properties of $\lambda'$ and $\lambda''$.
  \begin{itemize}
    \item By the naturality of $\lambda'$, 
    for each $(\sigma, a) \in \subdist(X) \times A$ and $\delta \in \nfafunc{Y}$, we have:
    \begin{equation} \label{eq:lambda'_nogo}
      \lambda'_{X, Y}(\sigma, a, \delta) = \lambda'_{\setn{1}\, , \setn{1}}\left(\Delta_{r}, \, a, \, \powfunc(!_Y + \{\checkmark\})^A(\delta) \right),
    \end{equation}
    where $r \coloneqq \sum_{x \in X}\sigma(x)$.
    \item For each $(\sigma, a) \in \subdist(X) \times A$,
    define a natural transformation $\rho\colon \powfunc \Rightarrow \erealfunc$
    by $\{\rho_{Y} \coloneqq \erealfunc(\pi_2) \circ \lambda''_{X, Y} \circ \langle \sigma,\, a, \,\Delta_{\powfunc(\kappa_1)(\_)} \rangle\}_Y$.
    By~\cref{eg:powfunc_natural_trans}, it is the constant natural transformation to $0$.
    Thus
    it follows that 
    $\lambda''_{X, Y}(\sigma, a, \Delta_{L})(x, y) = 0$ 
    for each $L \in \powfunc(X)$,
    $x \in X$, and $y \in Y$
    since 
    \[\rho_{Y}(L)(y) = \sum_{x \in X}\lambda''_{X, Y}(\sigma, a, \Delta_{L})(x, y) = 0.\]
  \end{itemize}
  
  By~\cref{prop:correctCriterion}, for each $(\sigma, a) \in \subdist(\mathbf{\Omega}_S) \times A$ and $\delta \in \nfafunc{\mathbf{\Omega}_R}$,
  \begin{align}
    \begin{split} \label{eq:nfa_correct}
    &\Big(1-\sum_{\mu \in \mathbf{\Omega}_S}\sigma(\mu)\Big) \cdot \delta(a)(\checkmark) + \sum_{w \in A^+} \Big(\sum_{\mu \in \mathbf{\Omega}_S} \sigma(\mu) \cdot \mu(w)\Big) \cdot \Big(\bigvee_{L \in \delta(a)}L(w)\Big) \\
    &= \lambda'_{\mathbf{\Omega}_S, \mathbf{\Omega}_R}(\sigma, a, \delta) + \sum_{\mu \in \mathbf{\Omega}_S, L \in \delta(a)}  \sum_{w \in A^+} \mu(w) \cdot L(w) \cdot \lambda''_{\mathbf{\Omega}_S, \mathbf{\Omega}_R}(\sigma, a, \delta)(\mu, L).
    \end{split}
  \end{align}
For simplicity, we write $\delta(a)$ and $L$ for their characteristic functions in the equation above.
  
  For each $r \in [0, 1]$, $a \in A$, and $\delta \in \nfafunc{\setn{1}}$,
  define $\sigma_1 \in \subdist(\mathbf{\Omega}_S)$ and $\delta_1 \in \nfafunc{\mathbf{\Omega}_R}$ by 
  \begin{align*}
    \sigma_1(x) \coloneqq 
    \begin{cases}
      r &\text{if }x = \Delta_0 \\
      0 &\text{otherwise}
    \end{cases}, \quad
    \delta_1 \coloneqq \powfunc(\Delta_{\emptyset}+\{\checkmark\})^A(\delta).
  \end{align*}
  Then \eqref{eq:nfa_correct} for $(\sigma_1, a)$ and $\delta_1$ 
  gives that
  $(1-r) \cdot \delta(a)(\checkmark) = \lambda'_{\mathbf{\Omega}_S, \mathbf{\Omega}_R}(\sigma_1, a, \delta_1) = \lambda'_{\setn{1}, \setn{1}}(\Delta_r, a, \delta)$.
  Therefore, 
    for each $(\sigma, a) \in \subdist(X) \times A$ and $\delta \in \nfafunc{Y}$,  \eqref{eq:lambda'_nogo} induces
    $
      \lambda'_{X, Y}(\sigma, a, \delta) 
      = \big(1-\sum_{x \in X}\sigma(x)\big) \cdot \delta(a)(\checkmark).
    $
  
  Let $r \in (0, 1]$, $w \in A^+$, $a \in A$, and $\mu \in \mathbf{\Omega}_S$ defined by $\mu(w) = 1$ and $\mu(w') = 0$ for each $w' \in A^+ \setminus \{w\}$.
  Define $\sigma_2 \in \subdist(\mathbf{\Omega}_S)$ and $\delta_2 \in \nfafunc{\mathbf{\Omega}_R}$ by 
  \begin{align*}
    \sigma_2(x) \coloneqq 
    \begin{cases}
      r &\text{if }x = \mu \\
      0 &\text{otherwise}
    \end{cases}, \quad
    \delta_2 \coloneqq 
    \Delta_{\{w\}}.
  \end{align*}
  Then \eqref{eq:nfa_correct} for $(\sigma_2, a)$ and $\delta_2$ gives that
  \begin{align*}
    r &= \sum_{w \in A^+} \Big(\sum_{\mu \in \mathbf{\Omega}_S} \sigma_2(\mu) \cdot \mu(w)\Big) \cdot \Big(\bigvee_{L \in \delta_2(a)}L(w)\Big) \\
    &= \sum_{\mu \in \mathbf{\Omega}_S, L \in \delta_2(a)}  \sum_{w \in A^+} \mu(w) \cdot L(w) \cdot \lambda''_{\mathbf{\Omega}_S, \mathbf{\Omega}_R}(\sigma_2, a, \delta_2)(\mu, L) = 0.
  \end{align*}
   This leads a contradiction, proving that no such distributive law $\lambda$ can exist.
\end{proof}
Note that this immediately implies the no-go theorem for 
distributive laws 
from $\subdist(\_) \times A$ and $\nfafunc{(\_)}$ to $\mcfuncproduct{(\_)}$
as well.

\section{Markov Chains and Multiset Finite Automata}
\label{sec:mcmfas}

In this section, we present a coalgebraic product construction of MCs and \emph{multiset finite automata (MFAs)}, and show the coalgebraic product construction is correct w.r.t.~the inference map that takes the expectation of the number of accepting paths. 
This immediately shows the correctness of the existing product construction of MCs and unambiguous finite automata~\cite{BenediktLW14,BaierK00023} as a special case. 
We further show that the coalgebraic product of MCs and MFAs is solvable in polynomial time.

\begin{definition}[semantic structure of MFAs]
    \label{def:mfas}
    An \emph{MFA} is a coalgebra $d\colon Y\rightarrow \mfafunc{Y}$, where the underlying set $Y$ is finite.
    The \emph{semantic structure $(\mathbf{\Omega}_{R}, \tau_R)$ of MFAs} is defined by (i) $\mathbf{\Omega}_{R} \defeq (\traceMFA, \preceq)$, where $\preceq$ is the pointwise order; 
    and (ii) $\tau_R\colon \mfafunc{\traceMFA}\rightarrow \traceMFA$ is given by
    \begin{align*}
        \tau_R(\delta)(w) \defeq \begin{cases}
            \delta(a)(\checkmark) &\text{ if } w = a,\\
            \sum_{\mu \in \traceMFA}\delta(a)(\mu)\cdot \mu(w')&\text{ if } w = a\cdot w',\\
            0 &\text{ otherwise.}
        \end{cases}
    \end{align*}
\end{definition}
The semantics of MFAs is the number of accepting paths for each word. 
MFAs are indeed a weighted automaton with the standard commutative ring $\nat$. 

\begin{definition} \label{def:dist_mfa}
    The \emph{distributive law} $\lambda_{X, Y}\colon \mcfunc{X}\times \mfafunc{Y}\rightarrow\mcmfafunc{X\times Y}$ is given by 
    \begin{align*}
        \lambda_{X, Y}(\sigma, a, \delta)(x, y) 
           &\defeq \delta(a)(y) \cdot \sigma(x), && \lambda_{X, Y}(\sigma, a, \delta)(\checkmark) 
           \defeq \delta(a)(\checkmark) \cdot \sigma(\checkmark).
    \end{align*}
\end{definition}
Importantly, the coalgebraic product $c\otimes_{\lambda} d\colon X\times Y\rightarrow \mcmfafunc{X\times Y}$ is not substochastic, that is, the sum $\sum_{(x', y')} (c\otimes_{\lambda} d)(x, y)(x', y') + (c\otimes_{\lambda} d)(x, y)(\checkmark)$ may be strictly greater than $1$, for some $(x, y)\in X\times Y$.

We now establish the correctness of the coalgebraic product for MCs and MFAs w.r.t.~the model checking problem that computes the expectation of the number of accepting paths.
\begin{proposition}[correctness]
    \label{prop:correctmcmfas}
    Consider semantic structures defined in \cref{eg:mc}, \cref{def:sem_product}, and \cref{def:mfas}.
    Then
    the coalgebraic product $c\otimes_{\lambda} d$ is correct
    w.r.t.~$q\colon \traceMC\times \traceMFA \rightarrow \exreal$ defined by 
    \begin{align*}
    q(\sigma, \mu) \defeq \sum_{w\in A^{+}} \mu(w) \cdot \sigma(w),
    &&\text{ for each $\sigma \in \traceMC$ and $\mu \in \traceMFA$.}
    \qed
    \end{align*}
\end{proposition}
\ifarxiv 
See~\cref{subsec:proofcorrect} for the proof.
\else 
See~\cite[Appendix B.1]{KWfull} for the proof.
\fi
As a direct consequence, the product of an MC $c$ and an unambiguous finite automaton $d$~\cite{BenediktLW14,BaierK00023} is correct w.r.t.~the inference map that computes the probability of paths that are accepting. 
In fact, $(\mu \Phi_d)(w) = 1$ iff $l$ is an accepting word, and $(\mu \Phi_d)(w) = 0$ otherwise, because $d$ is unambiguous,
implying that the model checking $q \circ (\mu \Phi_c \times \mu \Phi_d)$ precisely gives the probability of paths that are accepting. 

\begin{figure}
\centering
    \begin{subfigure}[b]{0.3\textwidth}
      \begin{tikzpicture}
        \node[draw,circle, fill=white, initial] (s1) {\scalebox{0.8}{$x, a$}};
        \node[draw,circle, fill=white, accepting, right= 1cm of s1 ] (s2) {\scalebox{0.5}{$\checkmark$}};
        \draw[->] (s1) -> node [above] {\scalebox{0.8}{$1/3$}} (s2);
        \draw[->] (s1) edge[loop above] node[right] {\scalebox{0.8}{$2/3$}} (s3); 
      \end{tikzpicture}
      \caption{\label{fig:MC}}
    \end{subfigure}
    \begin{subfigure}[b]{0.3\textwidth}
      \begin{tikzpicture}
        \node[draw,circle, fill=white, initial] (s1) {\scalebox{0.8}{$y_1$}};
        \node[draw,circle, fill=white, yshift=-1.5cm] (y2) {\scalebox{0.8}{$y_2$}};
        \node[draw,circle, fill=white, accepting, right= 1cm of s1 ] (s2) {\scalebox{0.5}{\checkmark}};
        \draw[->] (s1) -> node [above] {\scalebox{0.8}{$a$}} (s2);
        \draw[->] (s1) edge[loop above] node[right] {\scalebox{0.8}{$a$}} (s3);
        \draw[->] (s1) -> node[right] {\scalebox{0.8}{$a$}}  (y2) ; 
        \draw[->] (y2) edge[loop left] node[left] {\scalebox{0.8}{$a$}} (ylabel);
        \draw[->] (y2) -> node [below] {\scalebox{0.8}{$a$}} (s2);
      \end{tikzpicture}
      \caption{\label{fig:MFA}}
    \end{subfigure}
    \begin{subfigure}[b]{0.3\textwidth}
       \begin{tikzpicture}
          \node[draw,circle, fill=white, initial] (s1) {\scalebox{0.8}{$x, y_1$}};
          \node[draw,circle, fill=white, yshift=-1.5cm] (y2) {\scalebox{0.8}{$x, y_2$}};
          \node[draw,circle, fill=white, accepting, right= 1cm of s1 ] (s2) {\scalebox{0.8}{\checkmark}};
          \draw[->] (s1) -> node [above] {\scalebox{0.8}{$1/3$}} (s2);
          \draw[->] (s1) edge[loop above] node[right] {\scalebox{0.8}{$2/3$}} (s3);
          \draw[->] (s1) -> node[right] {\scalebox{0.8}{$2/3$}}  (y2) ; 
          \draw[->] (y2) edge[loop left] node[left] {\scalebox{0.8}{$2/3$}} (ylabel);
          \draw[->] (y2) -> node [below] {\scalebox{0.8}{$1/3$}} (s2);
      \end{tikzpicture}
      \caption{\label{fig:Product}}
    \end{subfigure}
    \caption{Examples: \subref{fig:MC} an MC, \subref{fig:MFA} an MFA, and \subref{fig:Product} their product.}
    \label{fig:examples}
\end{figure}
    
\begin{example}
  To illustrate the model checking problem with MFAs, we present a very simple example that is nevertheless rich enough to convey the core idea.
  Consider the MC and MFA shown in~\cref{fig:examples}. The MC models a server that repeatedly requests a response from a user until it reaches the state $\checkmark$, while the user’s behavior is modeled by the MFA. In the MFA, the target state $\checkmark$ represents an undesirable (bad) status for the user, which we would like to prevent the user from reaching as much as possible.

  The model checking problem in this setting does not ask for the probability of such undesirable user behaviors. Instead, it concerns the expected number of these behaviors. In this way, the analysis provides a quantitative measure of the safety guarantees on the user’s behavior.
The model checking result is $\big(q \circ (\mu \Phi_c \times \mu \Phi_d)\big)(x, y_1) = \sum^{\infty}_{n=1} (2/3)^{n-1}\cdot (1/3)\cdot n = 3$. 
This indicates that we should expect around three undesirable user behaviours until the server terminates. 
  The product of the MC and MFA is described in~\cref{fig:examples}: note that this is not an MC because $2/3+2/3+1/3 = 5/3 > 1$.
\end{example}

We remark that the distributive law $\lambda$ is the unique one:
\newcommand{\ev}{\mathrm{ev}}
\begin{proposition}[uniqueness under correctness criterion] \label{prop:lambda_uniqueness}
  The distributive law $\lambda$
  defined in~\cref{def:dist_mfa}
    is the unique one 
    satisfying the following properties.
    \begin{enumerate}
      \item There is a distributive law $\rho$
    from $\subdist$ and $\mfunc((\_)+1)$ to $\mcmfafunc{(\_)}$
    s.t.~$\lambda = \rho \circ (\id_{\subdist(\_)} \times \ev_{A, \mfunc((\_)+1)})$,
      where $\ev_{A, (\_)}\colon A \times \id^A \Rightarrow \id$ is given by evaluation maps.
      \item
    Any $c\otimes_{\lambda} d$ satisfies the correctness criterion in the setting of \cref{prop:correctmcmfas}.
    \qed
    \end{enumerate}
\end{proposition}
\ifarxiv 
See~\cref{subsec:proofunique} for the proof. 
\else     
See~\cite[Appendix B.2]{KWfull} for the proof. 
\fi 

\subsection{Model Checking of the Product}
Lastly, we study how to solve the product of MCs and MFAs. 
The product can be naturally regarded as an MC with \emph{multiplicative rewards}, which is very recently proposed in~\cite{BaierCMP25}.
This is because the number of accepting runs defined by MFAs can be seen as multiplicative rewards. 
Since we are only interested in traces that eventually reaches $\checkmark$, we can assume that the product MC with multiplicative rewards only has trivial bottom strongly connected components without loss of generality, and 
this is solvable in polynomial time as shown in~\cite[Theorem 3.18 and Remark 3.20]{BaierCMP25}. 
We summarise our complexity result as follows; in the following we assume that all probabilistic transitions in MCs are rational values. 
\begin{proposition}
  \label{prop:complexity}
  Given an MC $c\colon X\rightarrow\mcfunc{X}$, an MFA $d\colon Y\rightarrow \mfafunc{Y}$ and an initial state $(x, y)\in X\times Y$ of their product $c\otimes_{\lambda} d$, the semantics $\mu \Phi_{c\otimes_{\lambda} d}(x, y)$ is computable in polynomial time, implying that the model checking $q \circ (\mu \Phi_c \times \mu \Phi_d)(x, y)$ is solvable in polynomial time. 
\end{proposition}
\begin{proofs}
  The semantics of products can be seen as the expected multiplicative reward on the corresponding MC; 
  \ifarxiv 
  see~\cref{subsec:reduction} for the reduction. 
  \else    
  see~\cite[Appendix B.3]{KWfull} for the reduction. 
  \fi 
  By~\cite[Theorem 3.18 and Remark 3.20]{BaierCMP25}, there is a linear equation system such that (i) the semantics is a part of the unique non-negative solution of the linear equation system iff the linear equation system has a non-negative solution, and (ii) the semantics diverges ($\mu \Phi_{c\otimes_{\lambda} d}(x, y) = \infty$) iff there is no non-negative solution in the linear equation system. 
  Constructing and solving such a linear equation system can be done in polynomial time, concluding the proof. 
\end{proofs}

\section{Related Work}
\label{sec:related}

Product constructions for probabilistic systems, including MCs, have been extensively studied over the decades. 
In particular, efficient product constructions with $\omega$-regular automata have been a key focus of research~\cite{CouvreurSS03,BaierK00023,Vardi85,CourcoubetisY95,SickertEJK16}.  
Handling with $\omega$-regular properties in a unified manner within the coalgebraic framework remains a challenge, which we leave for future work. 
An intriguing direction for future research is to build upon existing studies on coalgebraic $\omega$-regular automata~\cite{UrabeSH16,CianciaV19} to extend our framework further. 

C{\^{\i}}rstea and Kupke~\cite{CirsteaK23} study product constructions of quantitative systems, including probabilistic systems, and nondeterministic parity automata.
They show that model checking can be reduced to solving nested fixed-point equations, while leaving the complexity analysis as future work.
In this paper, we present the first no-go theorem for product constructions, and we introduce a new polynomial-time algorithm for model checking with MCs and MFAs, reducing the problem to computation of expected multiplicative rewards on MCs~\cite{BaierCMP25}.
On the other hand, our framework does not support $\omega$-regular properties, which are the main focus of~\cite{CirsteaK23}.

No-go theorems for structures that enable computation by combining different semantics have been studied in the context of computational effects, particularly in relation to the distributive laws of monads. 
Varacca and Winskel, following a proof attributed to Plotkin, demonstrated the non-existence of distributive laws between the powerset monad and distribution monad~\cite{VaraccaW06}. 
Subsequently, Zwart and Marsden~\cite{ZwartM22} developed a unified theory for no-go theorem on distributive laws, encompassing known results such as those in~\cite{KlinS18,VaraccaW06}.
We aim to further investigate our no-go theorem on product constructions, exploring potential connections with these no-go theorems for distributive laws of monads. 
Notably, distributive laws do exist between the multiset monad and the distribution monad~\cite{Jacobs2021,KeimelP16,DahlqvistP018,DashStatonACT20}, suggesting that multisets interact well with distributions, as observed in~\cref{prop:lambda_uniqueness}.     

Dahlqvist and Neves~\cite{DahlqvistNeves} show a series of characterisation of natural transformations $\alpha\colon T^n \Rightarrow T$ for a certain type of functor $T$, including the finite powerset functor and the multiset functor. 
 Our characterisation covers natural transformations $\beta\colon F_A\Rightarrow F_B$, where $F_A$ and $F_B$ are different functors (see~\cref{tab:table_of_isomorphisms}).

\section{Conclusion}
\label{sec:conclusion}
We present a no-go theorem, demonstrating the incompatibility  of coalgebraic product constructions with MCs and NFAs. 
To establish this result, we develop a novel characterisation of natural transformations for certain functors. 
We then introduce a new model checking problem for MCs and MFAs, based on our coalgebraic product construction. 

As a future direction,
one could further
extend
the characterisation of natural transformations form $F_A$ to $F_B$ 
given in~\cref{thm:monoid_lambda},
since there are cases not covered even  
when $A$ is singly generated,
as well as cases 
when $A$ is finitely generated.

Another interesting direction is to study model checking of MCs with $\omega$-regular multiset automata.
Investigating the computational complexity of this problem—--possibly by applying the results in~\cite{BaierCMP25}---would be an interesting future work. 

\bibliography{main}
\ifarxiv
\newpage

\appendix

\section{Omitted Proofs in Section~\ref{sec:naturaltrans}}

\subsection{Proof of~\cref{prop:natural_trans_2}}
\label{subsec:proof_prop_natural_trans_2}
\begin{proof}
  Let $\lambda\colon F_A \Rightarrow F_B$ be a natural transformation.
  For each $f \in F_A(X)$ and $x \in X$,
  consider $g\colon X \to \setn{2}$ defined by $g(x) = 1$ and $g(x') = 2$ for each $x' \in X \setminus \{x\}$.
  By naturality of $\lambda$, we obtain
  $\lambda_X(f)(x) = \lambda_\setn{2}(F_A(g)(f))(1)$.
  Therefore $\lambda$ is determined by $\lambda_\setn{2}$.
\end{proof}

\subsection{Proof of faithfulness of $F_{(\_)}$} \label{ap:faithful}
\begin{proposition}
  The functor $F_{(\_)}$ is faithful.
\end{proposition}
\begin{proof}
  Consider arbitrary $i, j\colon A \to B$ in $\CMon$ such that $F_i = F_j$.
  For each $a \in A$, $i(a) = (F_i)_\setn{1}(\Delta_a) = (F_j)_\setn{1}(\Delta_a) = j(a)$.
\end{proof}

\subsection{Proof of~\cref{lem:monoid_lambda}}
\label{subsec:proof_lem_monoid_lambda}

\begin{proof}
  \eqref{item:f_equal} 
   If $x = x'$, the claim is trivially true.
   Otherwise,
   consider the function $g\colon X \to X$ defined by $g(x) \defeq x'$, $g(x') \defeq x$, $g(z) \defeq z$ for each $z \in X \setminus \{x, x'\}$.
   By naturality of $\lambda$ and $f = F_A(g)(f)$,
   it follows that $\lambda_X(f)(x) = \lambda_X(f)(x')$.
   
   \eqref{item:idempotent}
   For $n=0$, 
   suppose that 
   $f(x) = 0$.
   Consider any set $Y$ containing $X$,
   and let
   $f'\coloneqq F_A(\iota)(f) \in F_A(Y)$ where $\iota\colon X \hookrightarrow Y$. 
   As with $f$, $f'$ also satisfies $f'(x) = 0$ and $\lambda_Y(f')(x)$ is equal to $\lambda_X(f)(x)$ by naturality.
   Hence, we can assume that $|X|$ is sufficiently large,
   ensuring 
   the existence of $x' \in X \setminus \{x\}$ such that $\lambda_X(f)(x') = 0$
   since $\supp(\lambda_X(f))$ is finite.
   Define $g_1\colon X \to X$ by $g_1(x) \defeq x'$, $g_1(x') \defeq x'$, $g_1(z) \defeq z$ for each $z \in X\setminus \{x, x'\}$.
   By naturality of $\lambda$,
   we obtain $\lambda_X(f)(x) = \big(\lambda_X\circ F_A(g_1)\big)(f)(x) = \big(F_B(g_1)\circ \lambda_X\big)(f)(x) = 0$.
 
   For $n=1$, the statement is trivial.
   For $n \geq 2$,
   define $f'\in F_A(X+(\setn{2n-2}))$ by 
   $f' \defeq [f, \Delta_{f(x)}]$,
   and
   define $g_2\colon X+(\setn{2n-2}) \to X+(\setn{2n-2})$ by 
   $g_2(\kappa_1(x)) = g_2(\kappa_2(i)) \defeq 1$ for each $i \in \setn{n-1}$ and $g_2(z) \defeq z$ for other $z \in X+(\setn{2n-2})$
   where $\kappa_1$ and $\kappa_2$ are respectively first and second coprojections.
   By naturality of $\lambda$, we obtain 
   \begin{align} 
     \lambda_{X+(\setn{2n-2})}(F_A(g_2)(f'))(\kappa_2(1)) &= \lambda_{X+(\setn{2n-2})}(f')(\kappa_1(x)) + \sum_{i=1}^{n-1}\lambda_{X+(\setn{2n-2})}(f')(\kappa_2(i)), \label{eq:1i}\\
     \lambda_{X+(\setn{2n-2})}(F_A(g_2)(f'))(\kappa_2(n)) &= \lambda_{X+(\setn{2n-2})}(f')(\kappa_2(n)). \label{eq:nn}
   \end{align}
   By~\cref{lem:monoid_lambda}.\ref{item:f_equal}, we obtain
   \begin{align} 
     \lambda_{X+(\setn{2n-2})}(F_A(g_2)(f'))(\kappa_2(1)) &= \lambda_{X+(\setn{2n-2})}(F_A(g_2)(f'))(\kappa_2(n)), \label{eq:1n}\\
     \lambda_{X+(\setn{2n-2})}(f')(\kappa_1(x)) &= \lambda_{X+(\setn{2n-2})}(f')(\kappa_2(i)) &&\text{ for each }i \in \setn{\setn{2n-2}}. \label{eq:xi}
   \end{align}
   Therefore, it follows that 
   \begin{align*}
   \lambda_{X+(\setn{2n-2})}(f')(\kappa_2(n)) &= \lambda_{X+(\setn{2n-2})}(f')(\kappa_1(x)) + \sum_{i=1}^{n-1}\lambda_{X+(\setn{2n-2})}(f')(\kappa_2(i)) &&\text{by }\eqref{eq:nn}, \eqref{eq:1n}, \eqref{eq:1i},  \\
   &= n \cdot \lambda_{X+(\setn{2n-2})}(f')(\kappa_2(n)) &&\text{by }\eqref{eq:xi}.
   \end{align*}
   Then by naturality of $\lambda$ for $[\id, \Delta_x]\colon X+(\setn{2n-2}) \to X$, 
   noting that $n \cdot a = a$ implies $(2n-1) \cdot a = n \cdot a + (n-1) \cdot a = n \cdot a = a$,
   \begin{align*} 
     \lambda_{X}(f)(x) &= \lambda_{X+(\setn{2n-2})}(f')(\kappa_1(x)) + \sum_{i=1}^{2n-2} \lambda_{X+(\setn{2n-2})}(f')(\kappa_2(i)) &&\text{since }f = F_A([\id, \Delta_x])(f'),  \\
     &= (2n-1) \cdot \lambda_{X+(\setn{2n-2})}(f')(\kappa_2(n)) &&\text{by }\eqref{eq:xi},\\
     &= \lambda_{X+(\setn{2n-2})}(f')(\kappa_2(n)).
   \end{align*}
   Hence, it follows that $n \cdot \lambda_{X}(f)(x) = \lambda_{X}(f)(x)$.
 \end{proof}

 \subsection{Proof of~\cref{thm:monoid_lambda}}
 \label{appendix:proof_monoid_lambda}

 \begin{lemma} \label{lem:mlc}
  Let $C$ be a commutative monoid, $n \in \mathbb{N}_{> 1}$, and $c \in C$ be an element such that 
  $n \cdot c = c$.
  We write $[l]$ for the remainder of $l$ modulo $n-1$.
  Then the following statements hold.
  \begin{enumerate}
    \item \label{item:mlc} For each 
    $m, l \in \mathbb{N}_{\geq 1}$
    $[m] = [l]$ 
    implies $m \cdot c = l \cdot c$.
    \item If $n$ is the least natural number such that $n > 1$ and $n \cdot c = c$, 
    for each $m, l \in \mathbb{N}_{\geq 1}$,
    $m \cdot c = l \cdot c$ implies $[m] = [l]$.
  \end{enumerate}
 \end{lemma}
 \begin{proof}
 1) 
Assume $m \leq l$ without loss of generality.
 Since $[m] = [l]$, there exists $k \in \mathbb{N}$ such that $l = m+(n-1) \cdot k$.
 We show $m \cdot c = (m + (n-1) \cdot k) \cdot c$ by induction on $k$.
 If $k \leq m$, 
  $l \cdot c = (m-k+n\cdot k)\cdot c = (m-k) \cdot c + (n \cdot k) \cdot c = (m-k) \cdot c + k \cdot c = m \cdot c$.
  Otherwise, 
  $l \cdot c = (m+(n-1) k) \cdot c = (m+(n-1)(k-m) + (n-1) m)\cdot c = (m + (n-1) m)\cdot c = m \cdot c$ by induction hypothesis.

  2) 
  Assume
  $[m] \leq [l]$ without loss of generality.
  If $[m] \geq 1$, 
  by~\cref{lem:mlc}.\ref{item:mlc}, 
  it follows that $m \cdot c = [m] \cdot c$ and $l \cdot c = [l] \cdot c$, which implies $[m] \cdot c = [l] \cdot c$.
  Then it follows that $(n+[m]-[l]) \cdot c = (n+[l]-[l]) \cdot c = n \cdot c = c$.
  Since $n$ is the least natural number such that $n > 1$ and $n \cdot c = c$,
  we obtain $[m] = [l]$.

  It is clear that if $[m] = 0 = [l]$.
  If $[m] = 0$ and $[l] \geq 1$,
  we have $k \geq 1$ such that $m = (n-1)k \cdot c$ since $m \geq 1$.
  If $k > 1$, then $((n-1)\cdot (k-2) + (n-1) \cdot 2) \cdot c 
  = ((n-1)\cdot (k-2) + (n+n-2)) \cdot c 
  = ((n-1)\cdot (k-2) + (1+n-2)) \cdot c 
  = ((n-1)\cdot (k-1)) \cdot c$.
  Therefore, $m \cdot c = (n-1) \cdot c$ holds.
  By~\cref{lem:mlc}.\ref{item:mlc}, 
  $l \cdot c = [l] \cdot c$ holds, which implies $(n-1) \cdot c = [l] \cdot c$.
  Since $n$ is the least natural number such that $n > 1$ and $n \cdot c = c$,
  $[l] = 0$.
 \end{proof}

 \begin{proof}[Proof of~\cref{thm:monoid_lambda}]
  1) It is easy to see that $\lambda^b$ is natural,
  noting that
  $N(a' + a'') = N(a') + N(a'')$ for each $a', a'' \in A$.

  For each $n\in \mathbb{N}_{\geq 1}$, we use the function $f'_n \in F_A(\setn{n+2})$ given by $f'_n(i) \defeq a$ for each $i \in \setn{n}$ and $f'_n(n+1) = f'_n(n+2) \defeq 0$.
  Let us define the inverse of $\lambda^{(\_)}$.
  For $\lambda\colon F_A \Rightarrow F_B$,
  define $b\colon \mathbb{N} \to B$ by $b(0) = 0$ and $b(n) = \lambda_{\setn{n+2}}(f'_n)(1)$.

  We show that $\lambda^{(\_)}$ is a left inverse of $b^{(\_)}$, that is, $\lambda^{(\_)} \circ b^{(\_)} = \id$.
  For each $\lambda\colon F_A \Rightarrow F_B$,
  applying \cref{lem:lambda_m_n} for $\big(f, x, a, N(f(x)), \sum_{x' \in X \setminus \{x\}}N(f(x')), =, =\big)$, 
  we obtain 
  $\lambda_X(f)(x) 
  = N\big(f(x)\big) \cdot b\big(\sum_{x \in X} N(f(x))\big) 
  = \lambda^b_X(f)(x)$ for each $f \in F_A(X)$ and $x \in X$.

  We then show that $\lambda^{(\_)}$ is a right inverse of $b^{(\_)}$, that is, $b^{(\_)}\circ \lambda^{(\_)} = \id$.
  By definition of $f'_n$, for each $b \in B^{\mathbb{N}}$ with $b(0) = 0$
  and $n \in \mathbb{N}_{\geq 1}$,
  $\lambda^b_{\setn{n+2}}(f'_n)(1) = b(n)$.
  
  2)  
  We first show that $\lambda^b$ is natural.
  For each $g\colon X \to Y$, $f \in F_A(X)$, and $y \in Y$,
  \begin{align*}
    F_B(g)(\lambda^b_X(f))(y)
    &= \sum_{x \in g^{-1}(y)} N(f(x)) \cdot b\big([\sum_{x \in X}N(f(x))]\big)  \\
    &= \left(\sum_{x \in g^{-1}(y)} N(f(x))\right) \cdot b\Big(\big[\sum_{x \in X}N(f(x))\big]\Big), \\
    \lambda^b_Y(F_A(g)(f))(y)
    &= N(\sum_{x \in g^{-1}(y)} f(x)) \cdot b\Big(\big[\sum_{y \in Y}N(\sum_{x \in g^{-1}(y)} f(x))\big]\Big).
  \end{align*}
  By~\cref{lem:mlc},
  \begin{displaymath}
  [N(\sum_{x \in g^{-1}(y)} f(x))] = [\sum_{x \in g^{-1}(y)} N(f(x))] \text{ and } 
  \big[\sum_{x \in X}N(f(x))\big] = \big[\sum_{y \in Y}N(\sum_{x \in g^{-1}(y)} f(x))\big],
  \end{displaymath}
  and thus we also obtain 
  $F_B(g)(\lambda^b_X(f))(y) = \lambda^b_Y(F_A(g)(f))(y)$.

  Let us define the inverse of $\lambda^{(\_)}$.
 For $\lambda\colon F_A \Rightarrow F_B$,
  define $d^\lambda\colon \mathbb{N}_{\geq 1} \to \{c \in B \mid n \cdot c = c\}$ 
  and $b^\lambda\colon \{0, \cdots, n-2\} \to \{c \in B \mid n \cdot c = c\}$ 
  by 
  \begin{displaymath}
    d^\lambda(m) \coloneqq \lambda_{\setn{m+2}}(f'_m)(1),
    \quad
    b^\lambda(m) \coloneqq d^\lambda([m-1]+1),
  \end{displaymath}
  where $f'_m(i) = a$ for each $i \in \setn{m}$ and $f'_m(m+1) = f'(m+2) = 0$.
  By~\cref{lem:monoid_lambda}.\ref{item:idempotent},
  we have
  $n \cdot \lambda_{\setn{m+2}}(f'_m)(1) = \lambda_{\setn{m+2}}(f'_m)(1)$, which ensures $d^\lambda(m), b^\lambda(m) \in \{c \in B \mid n \cdot c = c\}$.

  We show that $\lambda^{(\_)}$ is a left inverse of $b^{(\_)}$, that is, $\lambda^{(\_)} \circ b^{(\_)} = \id$.
  Let $\lambda\colon F_A \Rightarrow F_B$ be any natural transformation.
  For each $m \in \mathbb{N}_{\geq 1}$,
   by \cref{lem:lambda_m_n} for $(f'_{m+n-1}, 1, a, 1, m-1, =, =)$, it follows that
  \begin{equation} \label{eq:bm}
    d^\lambda(m) = d^\lambda(m+n-1).
  \end{equation}
  Therefore, 
  \cref{lem:lambda_m_n} for $(f, x, a, N(f(x)), \sum_{x' \in X \setminus \{x\}}N(f(x')), =, =)$  implies that
  for each $f \in F_A(X)$ and $x \in X$,
  \begin{align*}
  \lambda_X(f)(x) 
  &= N\big(f(x)\big) \cdot d^\lambda\big(\sum_x N(f(x))\big)  \\
  &= N\big(f(x)\big) \cdot b^\lambda\big([\sum_x N(f(x))]\big) 
  = \lambda^{(b^\lambda)}_X(f)(x).
  \end{align*}
  The second equality holds because 
  $d^\lambda(\sum_x N(f(x))) = d^\lambda([\sum_x N(f(x))-1]+1) = b^\lambda\big([\sum_x N(f(x))]\big)$ by \eqref{eq:bm} if $\sum_{x \in X}N(f(x)) > 0$,
  and $N(f(x)) = 0$ if $\sum_{x \in X}N(f(x)) = 0$.
  
  We show that $\lambda^{(\_)}$ is a right inverse of $b^{(\_)}$, that is, $b^{(\_)}\circ \lambda^{(\_)} = \id$.
  Let $b\colon \{0, \cdots, n-2\} \to \{c \in B \mid n \cdot c = c\}$ be any function.
  Then for each $i \in \{0, \cdots, n-2\}$, 
  $b(i) = 1 \cdot b([i+n-1]) = \lambda^b_{\setn{i+n+1}}(f'_{i+n-1})(1) = d^{(\lambda^b)}(i+n-1) = b^{(\lambda^b)}(i)$.
  The last equality holds since 
  if $i=0$ then $d^{(\lambda^b)}(n-1) = b^{(\lambda^b)}(0)$,
  and if $i>0$ then
  $d^{(\lambda^b)}(i+n-1) = d^{(\lambda^b)}(i) = b^{(\lambda^b)}(i)$ by definition of $b^\lambda$ and \eqref{eq:bm}.
\end{proof}

\subsection{Proof of~\cref{lem:subfunctor}} \label{ap:proof_subfunctor}
\begin{proof}
  Define $\lambda' \colon F_A \Rightarrow F_B$ by $\lambda'_X(f) \defeq \lambda_X(f)$ if $f \in F(X)$
  and $\Delta_0$ otherwise.
  This $\lambda'$ satisfies the required commutativity.
\end{proof}

\subsection{Proof of~\cref{prop:subfunc_pm}} \label{ap:proof_subfunc_pm}
\begin{proof}[Proof of 2, 3, 4]
  2) The statement can be proved in the same way as~\cref{prop:subfunc_pm}.\ref{item:pd}.
  Note that the unique natural transformation is given by $\lambda_X(S)(x) = 0$.
  3)
  Suppose that a natural transformation $\lambda\colon \mfunc \Rightarrow \mathcal{D}$ exists.
  Since $\mathcal{D}$ is a subfunctor of $F_{([0, 1], +, 0)}$, 
  \cref{lem:subfunctor} and  \cref{lem:monoid_lambda}.\ref{item:idempotent} with $n=0$ imply that
  for any $X \in \sets$,
  $\lambda_X(\Delta_0)$ should be 
  the constant zero function, which is not contained in $\mathcal{D}(X)$. 
  Therefore,
  no natural transformation $\mfunc \Rightarrow \mathcal{D}$ exists.

  4) We define the inverse of $\lambda^{(\_)}$.
  Let $\lambda\colon \mfunc \Rightarrow \subdist$ be a natural transformation.
  By \cref{lem:subfunctor} and \cref{thm:monoid_lambda}.\ref{item:case_n}, there is a unique $b'\colon \mathbb{N} \to [0, 1]$ such that 
  $b'(0) = 0$ and
  for each $f \in \mfunc(X)$ and $x \in X$,
  $\lambda_X(f)(x) = f(x) \cdot b'\big(\sum_{x \in X}f(x)\big)$.
  Define $b\colon \mathbb{N} \to [0, 1]$ by $b(r) = r \cdot b'(r)$,
  noting that $\lambda_\setn{1}(\Delta_r)(1) \in \subdist(\setn{1})$ 
  ensures $r \cdot b'(r) \in [0, 1]$.
  It is easy to show that this mapping from $\lambda$ to $b$ is the inverse of $\lambda^{(\_)}$.
\end{proof}

\subsection{Proof of~\cref{prop:ereal_powfunc}} \label{ap:proof_ereal_powfunc}
\begin{proof}
  To show that $\lambda^{(\_)}$ is an isomorphism, we define its inverse
  $b^{(\_)}$ as follows: Given  $\lambda\colon \erealfunc \Rightarrow \powfunc$, define the function $b^{\lambda}\colon \nnegreal \to \bools$ by $b(0) \coloneqq 0$ and $b^{\lambda}(r) \coloneqq \lambda_\setn{1}(\Delta_{r})(1)$.

  We now prove that $\lambda^{(\_)}$ is the left inverse of $b^{(\_)}$. 
  Consider arbitrary $\lambda\colon \erealfunc \Rightarrow \powfunc$, $f \in \erealfunc(X)$, and $x \in X$.
  If $f(x) = 0$, the result 
  $\lambda_X(f)(x) = \lambda_X^{b^{\lambda}}(f)(x)$ follows trivially. Thus, we assume $f(x)  > 0$. 
  Then there exists $n, m \in \mathbb{N}$ such that $1 \leq n < m$ and $\frac{n}{m} \cdot \sum_{x' \in X}f(x') \leq f(x) \leq \frac{n+1}{m}\cdot \sum_{x' \in X}f(x')$.
  By \cref{lem:erealfunc_b}.\ref{item:b_lambda_1} when $B = (\bools, \lor, \bot)$,
  since all elements of $\bools$ are idempotent,
  there exists $b \in \mathbb{B}$ such that $\bigvee_{x' \in X \setminus \{x\}} \lambda_X(f)(x') = b \leq \lambda_X(f)(x)$.
  From this, we conclude that $\bigvee_{x' \in X} \lambda_X(f)(x') \leq \lambda_X(f)(x)$.
  Similarly,
  we also obtain $\lambda_X(f)(x) \leq \bigvee_{x' \in X}\lambda_X(f)(x')$.
  Since $\leq$ is a partial order on $\bools$, it follows that
  \begin{displaymath}
  \bigvee_{x' \in X}\lambda_X(f)(x') = \lambda_X(f)(x).
  \end{displaymath}
Moreover, the naturality of $\lambda$ for $!_X$ gives
$\bigvee_{x' \in X}\lambda_X(f)(x') = \lambda_\setn{1}(\Delta_{r})(1)$
where $r \defeq \sum_{x' \in X}f(x')$,
implying the equation $\lambda_X(f)(x) = \lambda_\setn{1}(\Delta_{r})(1) = \lambda_X^{b^{\lambda}}(f)(x)$. 

It is easy to prove that $\lambda^{(\_)}$ is the right inverse, completing the proof.
\end{proof}

\subsection{Proof of~\cref{cor:subdist_ereal}} \label{ap:proof_subdist_ereal}
\begin{proof}
  1, 2) Statement~\ref{item:subdist_1} and~\ref{item:subdist_2} can be proved in a similar manner.
  Here we only provide the proof for first one.
  It is clear that $\lambda^b$ is a natural transformation.

  We define the inverse of $\lambda^{(\_)}$ as follows.
  Let $\lambda$ be any natural transformation from $\subdist$ to $\erealfunc$.
  Since $\subdist$ is a subfunctor of $\erealfunc$,
  by \cref{lem:subfunctor} and \cref{prop:ereal_ereal}, there is a unique $b'\colon \nnegreal \to \nnegreal$ such that 
  $b'(0) = 0$,
  $\lambda_X(f)(x) = f(x) \cdot b'(\sum_{x \in X}f(x))$ 
  for each $f \in \subdist(X)$ and $x \in X$,
  and $b'(r) = 0$ for each $r > 1$.
  Define $b^\lambda\colon [0, 1] \to \nnegreal$ by restricting $b'$ to the domain $[0, 1]$.
  Then it is easy to show that this $b^{(\_)}$ is the inverse of $\lambda^{(\_)}$.

  3) We define the inverse of $\lambda^{(\_)}$.
  Let $\lambda\colon \subdist \Rightarrow \subdist$ be a natural transformation.
  By \cref{lem:subfunctor} and \cref{prop:ereal_ereal}, there is a unique $b'\colon \nnegreal \to \nnegreal$ such that 
  $b(0) = 0$,
  $\lambda_X(f)(x) = f(x) \cdot b'\big(\sum_{x \in X}f(x)\big)$
  and
  $\sum_{x \in X} f(x) \cdot b'\big(\sum_{x \in X}f(x)\big) \leq 1$
  for each $f \in \subdist(X)$ and $x \in X$,
  and $b'(r) = 0$ for each $r > 1$.
  Define $b^\lambda\colon [0, 1] \to [0, 1]$ by $b^\lambda(r) = r \cdot b'(r)$.
  It is easy to show that $b^{(\_)}$ is the inverse of $\lambda^{(\_)}$.
\end{proof}

\section{Omitted Proofs in Section~\ref{sec:mcmfas}}
\label{sec:proofsMcsMfas}

\subsection{Proof of~\cref{prop:correctmcmfas}}
\label{subsec:proofcorrect}
We prove that the product $c\otimes_{\lambda} d$ and the inference map $q$ satisfy the correctness criterion (\cref{prop:correctCriterion}). 
This is immediate by the following equations: 
\begin{align*}
  &(q\circ \tau_S\times \tau_R)(\nu, a, \delta) =  \sum_{w} \tau_S(\nu, a)(w)\cdot \tau_R(\delta)(w)\\
  &= \nu(\checkmark) \cdot \delta(a)(\checkmark) + \sum_{w} \big(\sum_{\sigma} \nu(\sigma)\cdot \sigma(w)\big)\cdot \big(\sum_{f} \delta(a)(f)\cdot f(w)\big),\\
  &(\tau_{S\otimes R}\circ F_{S\otimes R}(q)\circ \lambda_{\Omega_S, \Omega_R})(\nu, a, \delta) = \nu(\checkmark)\cdot \delta(a)(\checkmark) + \sum_{r} r\cdot \big((F_{S\otimes R}(q)\circ \lambda_{\Omega_S, \Omega_R})(\nu, a, \delta) \big)(r)\\
  &= \nu(\checkmark)\cdot \delta(a)(\checkmark) + \sum_{(\sigma, f)} q(\sigma, f)\cdot \big(\lambda_{\Omega_S, \Omega_R})(\nu, a, \delta) \big)(\sigma, f)\\
  &= \nu(\checkmark)\cdot \delta(a)(\checkmark) + \sum_{(\sigma, f)} q(\sigma, f)\cdot \nu(\sigma)\cdot \delta(a)(f)\\
  &= \nu(\checkmark)\cdot \delta(a)(\checkmark) + \sum_{(\sigma, f)} \big( \sum_{w}f(w)\cdot \sigma(w)\big) \cdot \nu(\sigma)\cdot \delta(a)(f)\\
  &= \nu(\checkmark)\cdot \delta(a)(\checkmark)+ \sum_{w} \big(\sum_{\sigma} \nu(\sigma)\cdot \sigma(w)\big)\cdot \big(\sum_{f} \delta(a)(f)\cdot f(w)\big). \qed
\end{align*}

\subsection{Proof of \cref{prop:lambda_uniqueness}}
\label{subsec:proofunique}
\begin{proof}
  Suppose that $\lambda\colon \times \circ \left((\subdist(\_) \times A)\times (\mfafunc{(\_)}) \right)\rightarrow\mcmfafunc{(\_)} \circ \times$
  is a distributive law satisfying the properties 1 and 2.
  Using property 2 and the natural isomorphism $\mcmfafunc{(\_)} \Rightarrow \erealfunc(\_) \times \nonnegreal$,
  there is a bijective correspondence between
  $\lambda$ and a pair of $\rho'\colon \times \circ \left(\subdist(\_)\times \mfunc{((\_)+\{\checkmark\})} \right)\rightarrow \exreal$
  and $\rho''\colon \times \circ \left(\subdist(\_)\times \mfunc{((\_)+\{\checkmark\})} \right)\rightarrow \erealfunc(\_) \circ \times$.
  We first describe key properties of $\rho'$ and $\rho''$.
  
  \begin{itemize}
    \item By naturality of $\rho'$, 
    for each $\sigma \in \subdist(X)$ and $\delta \in \mfunc(Y+\{\checkmark\})$,
    \begin{displaymath}
      \rho'_{X, Y}(\sigma, \delta) = \rho'_{\setn{1}, \setn{1}}\big(\Delta_{r}, \mfunc(!_Y+\{\checkmark\})(\delta)\big),
    \end{displaymath}
    where $r \coloneqq \sum_{x \in X}\sigma(x)$.
    \item For each $\sigma \in \subdist(X)$ and $n \in \mathbb{N}$,
    define a natural transformation $\alpha^{\sigma, n}\colon \mfunc \Rightarrow \erealfunc$
    by $\alpha^{\sigma, n}_{Y} = \erealfunc(\pi_2) \circ \rho''_{X, Y} \circ \langle \sigma, f_n \rangle$ where $f_n\colon \mfunc \Rightarrow \mfunc((\_)+\{\checkmark\})$ defined by $(f_n)_Y(d) = [d, \Delta_n]$ for each $d \in \mfunc(Y)$.
    By~\cref{eg:mfunc_natural_trans}, 
    there exists a funciton $b_{\sigma, n}\colon \mathbb{N}_{\geq 1} \to \nnegreal$ such that 
    \begin{equation} \label{eq:mfa_alpha}
      \alpha^{\sigma, n}_Y(d)(y) =
      d(y) \cdot b_{\sigma, n}\big(\sum_{y \in Y} d(y)\big),
    \end{equation}
    where $b_{\sigma, n}(0) = 0$.

    Next, for each $\delta \in \mfunc(Y+\{\checkmark\})$ and $y \in Y$,
    define a natural transformation $\beta^{\delta, y}\colon \subdist \Rightarrow \erealfunc$
    by
    $\beta^{\delta, y}_X = \gamma^y \circ \rho''_{X, Y} \circ \langle \id, \delta \rangle$
    where $\gamma^y\colon \erealfunc((\_) \times Y) \Rightarrow \erealfunc$ is given by $\gamma^y_X(g) = g \circ \langle \id, \Delta_y \rangle$ for each $g \in \erealfunc(X \times Y)$.
    Then~\cref{cor:subdist_ereal}.\ref{item:subdist_1} implies that 
    there is $b'_{\delta, y}\colon [0, 1] \to \nnegreal$ with $b'_{\delta, y}(0) = 0$ such that
    \begin{equation} \label{eq:beta_b}
    \beta^{\delta, y}_X(\sigma)(x) 
    = \sigma(x) \cdot b'_{\delta, y}\big(\sum_{x \in X}\sigma(x)\big).
    \end{equation}
    Then the following equalities hold:
    \begin{align*}
      \sum_{x \in X}\sigma(x) \cdot b'_{\delta, y}(\sum_{x \in X}\sigma(x))
      &= \sum_{x \in X}\beta^{\delta, y}_X(\sigma)(x) &&\text{by }\eqref{eq:beta_b} \\
      &= \sum_{x \in X}\rho''_{X, Y}(\sigma, \delta)(x, y) &&\text{by definition of }\beta \\
      &= \alpha^{\sigma, \delta(\checkmark)}_Y(\delta \circ \kappa_1)(y) &&\text{by definition of }\alpha\\
      &= \delta(y) \cdot b_{\sigma, \delta(\checkmark)}\big(\sum_{y \in Y} \delta(y)\big). &&\text{by }\eqref{eq:mfa_alpha}
    \end{align*}
    Thus, 
    for each $\sigma \in \subdist(X)$, $\delta \in \mfunc(Y+\{\checkmark\})$, $x \in X$, and $y \in Y$,
    one deduces that
    \begin{displaymath}
      \rho''_{X, Y}(\sigma, \delta)(x, y) 
      = \beta^{\delta, y}_X(\sigma)(x)
      =
      \begin{cases}
      \frac{\sigma(x)}{\sum_{x \in X}\sigma(x)} \cdot \delta(y) \cdot b_{\sigma, \delta(\checkmark)}(\sum_{y \in Y} \delta(y)) &\text{if }\sum_{x \in X}\sigma(x) > 0, \\
      0 &\text{otherwise}.
      \end{cases}
    \end{displaymath}

    Applying naturality of $\rho''$ for $!_X$ and $!_Y$, we find that
    for each $\sigma \in \subdist(X)$ and $\delta \in \mfunc(Y+\{\checkmark\})$ such that  $\sum_{x \in X}\sigma(x) > 0$,
    \begin{align*}
    \sum_{y \in Y} \delta(y) \cdot b_{\sigma, \delta(\checkmark)}(\sum_{y \in Y} \delta(y))
    &= \sum_{x \in X, y \in Y}\rho''_{X, Y}(\sigma, \delta)(x, y) \\
    &= \rho''_{\mathbf{1}, \mathbf{1}}(\Delta_{r}, (1 \mapsto \sum_{y \in Y}\delta(y), \checkmark \mapsto \delta(\checkmark)))(1, 1) \\
    &= \sum_{y \in Y}\delta(y) \cdot b_{\Delta_{r}, \delta(\checkmark)}\big(\sum_{y \in Y} \delta(y)\big),
    \end{align*}
    where $r \coloneqq \sum_{x \in X}\sigma(x)$.
    Therefore, noting that $b_{\sigma, \delta(\checkmark)}(0) = b_{\Delta_{r}, \delta(\checkmark)}(0)$,
    it follows that $b_{\sigma, \delta(\checkmark)}(\sum_{y \in Y} \delta(y)) = b_{\Delta_{r}, \delta(\checkmark)}(\sum_{y \in Y} \delta(y))$, so that we obtain 
    \begin{displaymath}
    \rho''_{X, Y}(\sigma, \delta)(x, y) = \frac{\sigma(x)}{\sum_{x \in X}\sigma(x)} \cdot \delta(y) \cdot b_{\Delta_{r}, \delta(\checkmark)}\big(\sum_{y \in Y} \delta(y)\big)
    \end{displaymath}
    if $\sum_{x \in X}\sigma(x) > 0$.
  \end{itemize}

  By the correctness, for each $(\sigma, a) \in \subdist(\mathbf{\Omega}_S) \times A$ and $\delta \in \mfafunc{\mathbf{\Omega}_R}$,
  \begin{align}
    \begin{split} \label{eq:mfa_correct}
    &\Big(1-\sum_{\mu \in \mathbf{\Omega}_S}\sigma(\mu)\Big) \cdot \delta(a)(\checkmark) + \sum_{w \in A^+} \Big(\sum_{\mu \in \mathbf{\Omega}_S} \sigma(\mu) \cdot \mu(w)\Big) \cdot \Big(\sum_{\delta' \in \mathbf{\Omega}_R}\delta(a)(\delta') \cdot \delta'(w)\Big) \\
    &= \rho'_{\mathbf{\Omega}_S, \mathbf{\Omega}_R}(\sigma, \delta(a)) + \sum_{\mu \in \mathbf{\Omega}_S, \delta' \in \mathbf{\Omega}_R}  \sum_{w \in A^+} \mu(w) \cdot \delta'(w) \cdot \rho''_{\mathbf{\Omega}_S, \mathbf{\Omega}_R}(\sigma, \delta(a))(\mu, \delta').
    \end{split}
  \end{align}
  
  For each $r \in [0, 1]$, $n_1, n_2 \in \mathbb{N}$,
  define $\sigma_1 \in \subdist(\mathbf{\Omega}_S)$ and $\delta_1 \in \mfafunc{\mathbf{\Omega}_R}$ by 
  \begin{align*}
    \sigma_1(x) \coloneqq 
    \begin{cases}
      r &\text{if }x = \Delta_0, \\
      0 &\text{otherwise,}
    \end{cases}&& 
    \delta_1(a')(i) \coloneqq 
    \begin{cases}
      n_1 &\text{if } i=\Delta_0, \\
      0 &\text{if } i \in \mathbf{\Omega}_R \setminus \{\Delta_0\}, \\
      n_2 &\text{if }i=\checkmark.
    \end{cases}
  \end{align*}
  Then \eqref{eq:mfa_correct} for $(\sigma_1, a)$ and $\delta_1$ 
  yields 
  $(1-r) \cdot n_2 = \rho'_{\mathbf{\Omega}_S, \mathbf{\Omega}_R}(\sigma_1, \delta_1(a)) = \rho'_{\setn{1}, \setn{1}}(\Delta_r, (1 \mapsto n_1, \checkmark \mapsto n_2))$.
  Therefore, $\rho'_{X, Y}(\sigma, \delta) = (1-\sum_{\mu \in \mathbf{\Omega}_S}\sigma(\mu)) \cdot \delta(\checkmark)$ for each $\sigma \in \subdist(X)$ and $\delta \in \mfunc(Y+1)$.
  
  For each $r \in (0, 1]$, $a \in A$, and $n_1, n_2 \in \mathbb{N}$ $(n_2 \geq 1)$,
  consider a word $w \in A^+$ and $\mu \in \mathbf{\Omega}_S$ defined by $\mu(w) = 1$ and $\mu(w') = 0$ for each $w' \in A^+ \setminus \{w\}$, and
  define 
  $\sigma_2 \in \subdist(\mathbf{\Omega}_S)$ and
  $\delta_2 \in \mfafunc{\mathbf{\Omega}_R}$ by 
  \begin{align*}
    \sigma_2(x) \coloneqq 
    \begin{cases}
      r &\text{if }x = \mu, \\
      0 &\text{otherwise,}
    \end{cases} 
    &&\delta_2(a')(i) \coloneqq 
    \begin{cases}
      n_1 &\text{if }a' = a \text{ and }i = \checkmark,  \\
      n_2 &\text{if }a'=a \text{ and }i=\delta', \\
      0 &\text{otherwise,}
    \end{cases}
  \end{align*}
  where $\delta' \in \mathbf{\Omega}_R$ is defined by $\delta'(w) = 1$ and $\delta'(w') = 0$ for each $w' \in A^+ \setminus \{w\}$.
  Then \eqref{eq:nfa_correct} for $(\sigma_2, a)$ and $\delta_2$ gives that
    $r \cdot n_2 
    = n_2 \cdot b_{\Delta_r, n_1}(n_2)$.
    Hence, $b_{\Delta_r, n_1}(n_2) = r$, and consequently,
      $\rho''_{X, Y}(\sigma, \delta)(x, y) = 
        \sigma(x) \cdot \delta(y)$  for each $\sigma \in \subdist(X)$ and $\delta \in \mfunc(Y+1)$.
        
  The distributive law $\lambda$ given by the pair $(\rho', \rho'')$ is the one defined in \cref{def:dist_mfa}. This establishes the existence and uniqueness of $\lambda$.
\end{proof}

\subsection{Reduction}
\label{subsec:reduction}

Let $c\colon X\rightarrow\mcfunc{X}$ be an MC with rational transition probabilities and $d\colon Y\rightarrow \mfafunc{Y}$ be an MFA. 
We construct an MC $e\colon X\times Y\rightarrow \mathcal{D}(X\times Y+\{\checkmark\})\times \nonnegrat$ from the product $c\otimes_{\lambda} d\colon X\times Y\rightarrow \mcmfafunc{X\times Y}$. 
Indeed, the construction of $e$ is very simple: The transition probabilities $P$ of $e$ are given by $P\big((x, y), z\big) \defeq \frac{(c\otimes_{\lambda} d)(x, y)(z)}{\sum_{z'} (c\otimes_{\lambda} d)(x, y)(z')}$.
Here we can assume $\sum_{z'} (c\otimes_{\lambda} d)(x, y)(z')> 0$ without loss of generality,  because if it is not the case then we create an additional sink state with reward $0$, and add an transition to the sink with the probability $1$. 
The reward function $R$ of $e$ is given by $R(x, y) \defeq \sum_{z} (c\otimes_{\lambda} d)(x, y)(z)$.

We show that the semantics of the product $c\otimes_{\lambda} d$ is precisely the expected multiplicative rewards over the traces that eventually reaches $\checkmark$ on $e$. 
Indeed, by the straightforward induction we can see that the semantics of $e$ defined by the following semantic structure $([0, \infty], \tau_e)$ is the expected multiplicative rewards: the modality $\tau_e \colon \mathcal{D}([0, \infty]+\{\checkmark\})\times \nonnegrat\rightarrow [0, \infty]$ is given by 
\[
\tau_e(\nu, q) \defeq q\cdot \big( \nu(\checkmark) + \sum_{r}\nu(r)\cdot r\big).
\] 
By the construction of $e$, its semantics is precisely the semantics of the product.

\else
\fi

\end{document}